%% file: main.tex
\documentclass{article}
\input{formatting_examples/ijcai2023/ijcai23_header.tex}
\input{common/includes}
\input{common/definitions}
\usepackage{amsmath}

\title{Game Theory with Simulation of Other Players}

\author{
Vojtěch Kovařík
\and
Caspar Oesterheld\And
Vincent Conitzer\\
\affiliations
Foundations of Cooperative AI Lab (FOCAL),
Carnegie Mellon University\\
\emails
vojta.kovarik@gmail.com,
oesterheld@cmu.edu,
conitzer@cs.cmu.edu
}

\begin{document}

\maketitle

\begin{abstract}
    \input{abstract}
\end{abstract}

\section{Introduction}\label{sec:intro}
    \input{figures/simgames_efg.tex}
    \input{intro}
    \input{illustrative_examples}
    \input{outline}

\section{Background}\label{sec:background}
    \input{background}

\section{Simulation Games}\label{sec:simgame-basics}
    \input{theory_basic_properties}

\section{Structural Properties}\label{sec:structural}
    \input{theory_structural.tex}

\section{Computational Aspects}\label{sec:complexity}
    \input{theory_computational.tex}

\section{Effects on Players' Welfare}\label{sec:theory:social-welfare}
    \input{theory_social_welfare.tex}

    \input{literature}
\section{Discussion}\label{sec:discussion}
    \input{conclusions}

\anonymise{
    \input{acknowledgements}
}{}

\bibliographystyle{named}
    \bibliography{main}

\clearpage
\appendix

\draftOnly{
    \input{todos.tex}
}

\input{proofs.tex}
\input{piecewise_linearity.tex}

\end{document}

%% file: formatting_examples/ijcai2023/ijcai23_header.tex
\usepackage{ijcai23}

\usepackage{times}
\usepackage{soul}
\usepackage{url}
\usepackage[hidelinks]{hyperref}
\usepackage[utf8]{inputenc}
\usepackage[small]{caption}
\usepackage{graphicx}
\usepackage{amsmath}
\usepackage{amsthm}
\usepackage{booktabs}
\usepackage{algorithm}
\usepackage{algorithmic}
\usepackage[switch]{lineno}

%% file: common/includes.tex
\usepackage[utf8]{inputenc}
\usepackage{amsmath,amsfonts,amssymb,amsthm}
	\numberwithin{equation}{section} 
\usepackage[table]{xcolor}
\usepackage{enumitem}[shortlabels]
\usepackage{thm-restate}
\usepackage[T1]{fontenc}    
\usepackage{hyperref}       
    \usepackage{url}            
\usepackage{cleveref}       
\usepackage{booktabs}       
\usepackage{nicefrac}       
\usepackage{microtype}      
\usepackage{todonotes}
\usepackage{graphbox}       

\usepackage{ifthen}

%% file: common/definitions.tex
\newboolean{commentsactivated}
\setboolean{commentsactivated}{false}
\newcommand{\vc}[1]{\ifthenelse{\boolean{commentsactivated}}{{\color{blue} {\em VC: #1 }}}{}}
\newcommand{\co}[1]{\ifthenelse{\boolean{commentsactivated}}{{\color{teal} {\em CO: #1 }}}{}}
\newcommand{\vk}[1]{\ifthenelse{\boolean{commentsactivated}}{{\color{red} {\em VK: #1 }}}{}}
\newcommand{\lh}[1]{\ifthenelse{\boolean{commentsactivated}}{{\color{brown} {\em LH: #1 }}}{}}
\newcommand{\cs}[1]{\ifthenelse{\boolean{commentsactivated}}{{\color{pink} {\em CS: #1 }}}{}}
\newcommand{\ns}[1]{\ifthenelse{\boolean{commentsactivated}}{{\color{orange} {\em CS: #1 }}}{}}

\newcommand{\draftOnly}[1]{\ifthenelse{\boolean{commentsactivated}}{#1}{}}

\newboolean{highlightEdits}
\setboolean{highlightEdits}{false}
\newcommand{\edit}[1]{\ifthenelse{\boolean{highlightEdits}}{{\color{purple}{#1}}}{#1}}

\newboolean{journalTextShown}
\setboolean{journalTextShown}{true}
\newcommand{\journal}[1]{\ifthenelse{\boolean{journalTextShown}}{{{\ifthenelse{\boolean{highlightEdits}}{{\color{olive}{#1}}}{#1}}}}{}}

\newboolean{anonymous}
\setboolean{anonymous}{false}
\newcommand{\anonymise}[2]{\ifthenelse{\boolean{anonymous}}{{{#2}}}{#1}}

\newtheorem{definition}{\protect\definitionname}
\newtheorem{lemma}[definition]{\protect\lemmaname}
\newtheorem{proposition}[definition]{\protect\propositionname}

\newtheorem{observation}[definition]{\protect\observationname}

\theoremstyle{definition}           
    
    \newtheorem{example}[definition]{\protect\examplename}
    \newtheorem{notation}[definition]{\protect\notationname}

\providecommand{\corollaryname}{Corollary}
\providecommand{\claimname}{Claim}
\providecommand{\definitionname}{Definition}
\providecommand{\lemmaname}{Lemma}
\providecommand{\notationname}{Notation}
\providecommand{\remarkname}{Remark}
\providecommand{\problemname}{Problem}
\providecommand{\propositionname}{Proposition}
\providecommand{\examplename}{Example}
\providecommand{\theoremname}{Theorem}
\providecommand{\conjecturename}{Conjecture}
\providecommand{\observationname}{Observation}

\DeclareMathAlphabet{\mathpzc}{OT1}{pzc}{m}{it}
\DeclareMathSymbol{\shortminus}{\mathbin}{AMSa}{"39}

\newcommand{\defword}[1]{\textbf{\boldmath{#1}}}

\newcommand{\N}{\mathbb{N}}
\newcommand{\R}{\mathbb{R}}
\newcommand{\E}{\mathbf{E}}
\newcommand{\mc}{\mathcal}

\newcommand{\symbolPlaceholder}{{\, \cdot \, }}

\newcommand{\actions}{\mc A}
\newcommand{\action}{a}
\newcommand{\actionAlt}{{\action'}}
\newcommand{\actionOpp}{b}
\newcommand{\actionOppAlt}{{\actionOpp'}}
\newcommand{\actionSetOpp}{B}
\newcommand{\policy}{\pi}

\newcommand{\policies}{\Pi}

\newcommand{\pl}{i}

\newcommand{\opp}{\others}          
\newcommand{\others}{{\textnormal{-}\pl}}

\newcommand{\Plone}{\textnormal{P1}}
\newcommand{\Pltwo}{\textnormal{P2}}

\newcommand{\utility}{u}
\newcommand{\br}{\textnormal{br}}                     
\newcommand{\NE}{\textnormal{NE}}                     
\newcommand{\pureSE}{\textnormal{SE}_{\textnormal{pure}}^{\textnormal{\Pltwo{}}}}   

\newcommand{\supp}{\textnormal{supp}}                

\newcommand{\game}{{\mc{G}}}
\newcommand{\gameAlt}{{\game'}}

\DeclareMathSymbol{\shortminus}{\mathbin}{AMSa}{"39}

\newcommand{\actionSubset}{S}
    \newcommand{\simcost}{\texttt{c}}
\newcommand{\simSubscript}{{\textnormal{sim}}}
\newcommand{\simgame}{\game_\simSubscript}

\newcommand{\all}{\textnormal{all}}
\newcommand{\simulate}{\texttt{S}}
\newcommand{\brSelector}{\psi_\br}
\newcommand{\brSelectorPure}{\brSelector^\pure}
\newcommand{\pure}{\textnormal{det}}
\newcommand{\brSelectors}{{\Psi}_\br}
\newcommand{\brSelectorsPure}{\brSelectors^\pure}
\newcommand{\NEtrajectory}{\mathpzc{t}}
\newcommand{\exceptionPoint}{e}
\newcommand{\exceptionSet}{E}

\newcommand{\slackVariable}{w}
\newcommand{\LPmatrix}{A}
\newcommand{\LPrhs}{y^\simcost}
\newcommand{\LPvariables}{x_\simcost}

\newcommand{\transpose}{\textbf{T}}
\newcommand{\indices}{\mc I}

\newcommand{\basis}{B}
\newcommand{\regularBases}{\mc B}

\newcommand{\VoI}{\textrm{VoI}_\simulate}

\newcommand{\slope}{\alpha}
\newcommand{\simProb}{p}
\newcommand{\baseline}{\textrm{B}}
\newcommand{\deviate}{\textrm{D}}
\newcommand{\deviation}{d}

\newcommand{\bigO}{O}
\newcommand{\guessTheNumberConstant}{N}

\newcommand{\Alice}{\textrm{A}}
\newcommand{\Bob}{\textrm{B}}
\newcommand{\trustGame}{\textnormal{TG}}
\newcommand{\simTrustGame}{\trustGame_\simSubscript}
\newcommand{\trust}{\textnormal{T}}
\newcommand{\walkOut}{\textnormal{WO}}
\newcommand{\cooperate}{\textnormal{C}}
\newcommand{\defect}{\textnormal{D}}

\newcommand{\coopUtil}{50}
\newcommand{\coopUtilAlice}{25}
\newcommand{\coopUtilBob}{25}
\newcommand{\defectLossAlice}{100}      
\newcommand{\defectUtilBob}{150}
\newcommand{\utilDiffAlice}{125}   
\newcommand{\simcostExample}{7}     
\newcommand{\eqUtilAlice}{16.25}        
\newcommand{\eqUtilBob}{\coopUtilBob{}}
\newcommand{\secondBreakpoint}{20}      

\newcommand{\direction}{\delta}
\newcommand{\deviationRatio}{r}
\newcommand{\deviationRatioInverse}{r^{-1}}

%% file: abstract.tex
Game-theoretic interactions with AI agents could differ from traditional human-human interactions in various ways.
One such difference is that it may be possible to simulate an AI agent (for example because its source code is known),
    which allows others to accurately predict the agent's actions.
    This could lower the bar for trust and cooperation.
In this paper, we formalize
    games in which one player can simulate another at a cost.
We first derive some basic properties of such games
and then prove a number of results for them, including:
    (1) introducing simulation into generic-payoff normal-form games makes them easier to solve;
    (2) if the only obstacle to cooperation is a lack of trust in the possibly-simulated agent,
        simulation enables equilibria that improve the outcome for both agents; and
    however (3) there are settings where introducing simulation results in strictly worse outcomes for both players.

%% file: figures/simgames_efg.tex
\begin{figure*}[!tb]
    \centering
    \includegraphics[width=0.3\textwidth]{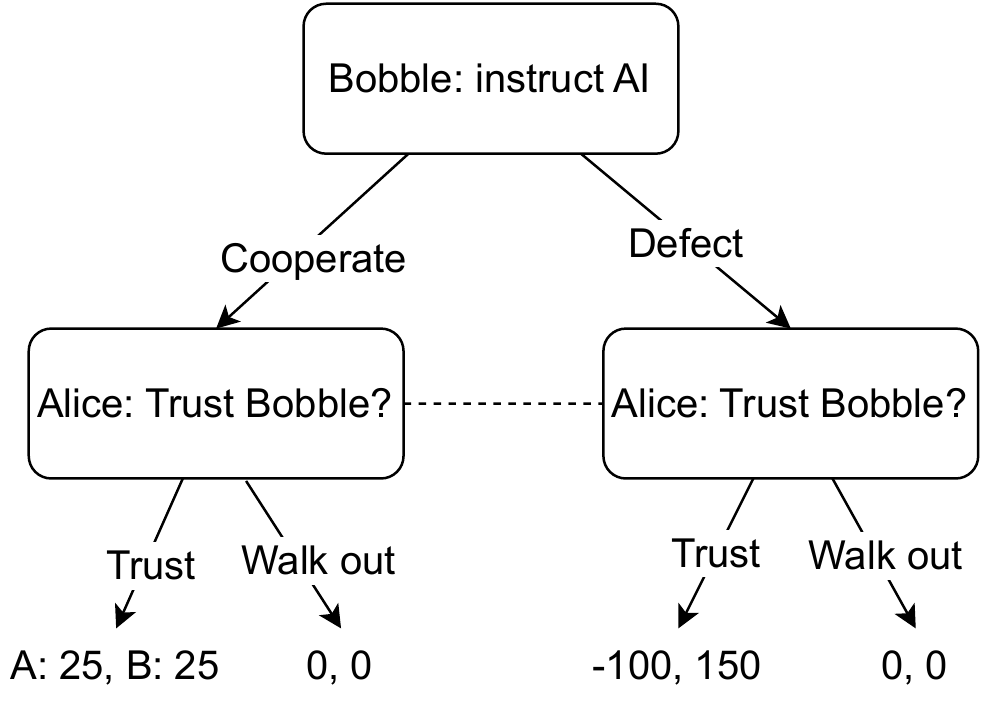}
    \hfill
    \includegraphics[width=0.6\textwidth]{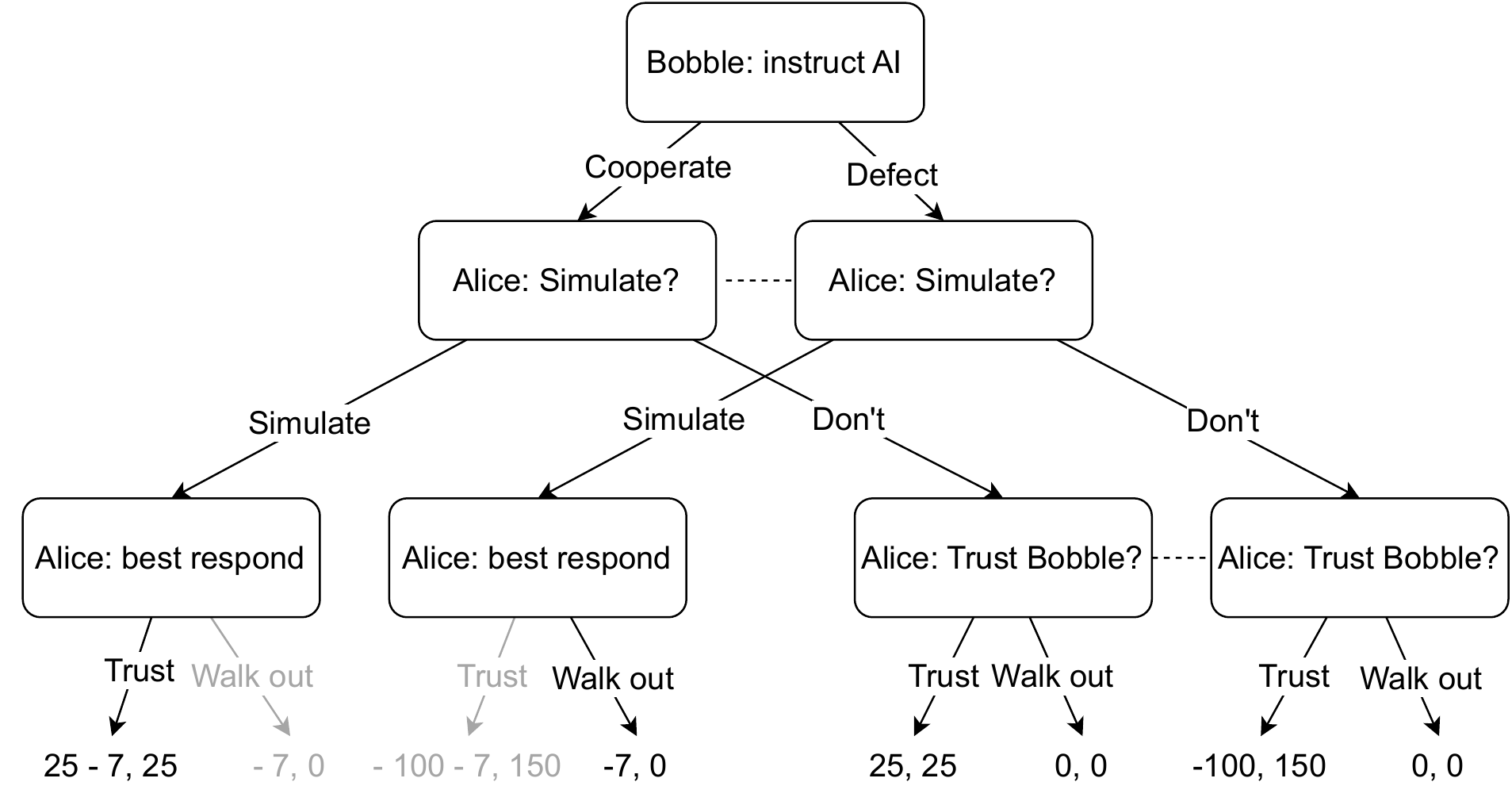}
    \caption{The underlying trust game $\trustGame$ (left) and the corresponding simulation game $\simTrustGame$ (right).
    }
    \label{fig:trust-game-efg}
\end{figure*}

%% file: intro.tex
Game theory is in principle agnostic as to the nature of the players: besides individual human beings, they can be households, firms, countries, and indeed AI agents.  Nevertheless, throughout most of the development of the field, game theorists have had in mind players that were either humans or entities whose decisions were taken by humans; and as with any theory, the examples one has in mind while developing that theory are likely to affect its focus.  If we try to re-develop game theory specifically with AI agents in mind, how might the theory turn out different?  Of course, theorems in traditional game theory will not suddenly become false just because of the change in focus.  Instead, we would expect any difference to consist in the kinds of settings and phenomena for which we develop models, analysis, and computational tools.

In this paper, we focus on one specific phenomenon that is more pertinent in the context of AI agents: agents being able to {\em simulate} each other.  If an agent's source code is available, another agent can simulate what the former agent will do, which intuitively appears to significantly change the game strategically.
We consider settings in which one agent can simulate another, and if they do so, they learn what the other agent will do in the actual game;
    however, simulating comes at a cost to the simulator, and therefore it is not immediately clear whether and when simulation will actually be used in equilibrium.
In particular, we are interested in understanding whether and when the availability of such simulation results in play that is more cooperative.
    For example, in settings where {\em trust} is necessary for cooperative behavior~\cite{berg1995trust}, one may expect that the ability to simulate the other player can help to establish this trust.
    But does this in fact happen in equilibrium?
    And if so, does the ability to simulate foster cooperation in all games, or are there games where it backfires?
    Are we even able to compute equilibria of games with the ability to simulate?

\medskip

In terms of related work,
    our setting is similar to the one of credible commitment \cite{von1934marktform},
        except that one needs to decide whether to pay for allowing the \textit{other} player to commit.
    Another perspective is that we study program equilibria \cite{Tennenholtz2004},
        except that only \textit{one} player's program can read the other's source code, and has to pay a cost to do so. 
    For further discussion and references, see Section~\ref{sec:related-work}.

\medskip

In the remainder of this introduction,
    we describe a specific example of a trust game and use it to overview the technical results presented later.
    We also give several examples that illustrate how simulation can lead to different results when moving beyond trust games.
    \textbf{\textit{For a quick overview, the key takeaways are in \Cref{sec:illustrative-examples}, highlighted in italics.}}

%% file: illustrative_examples.tex
\subsection{Overview and Illustrative Examples}\label{sec:illustrative-examples}

\paragraph{Trust Game}
As a motivation,
    consider the following Trust Game
    (cf. Figure\,\ref{fig:trust-game-efg}; our $\trustGame$ is a variation on the traditional one from \cite{berg1995trust}).
Alice has \$\defectLossAlice{}k in savings,
    which are currently sitting in her bank account at a 0\% interest rate.
She is considering hiring an AI assistant from the company Bobble
    to manage her savings instead.
    If Bobble and its AI \textit{cooperate} with her, the collaboration generates a profit of \$\coopUtil{}k,
    to be split evenly between her and Bobble.
However, Alice is reluctant to \textit{trust} Bobble,
    which might have instructed the AI to \textit{defect} on Alice by pretending to malfunction, while siphoning off all of the \$\defectUtilBob{}k.
In fact, the only Nash equilibria of this scenario are ones where
    Bobble defects on Alice with high probability, 
    and Alice, expecting this, \textit{walks out} on Bobble.

\paragraph{Adding simulation} 
Dismayed by their inability to make a profit,
    Bobble decides to share with Alice a portion of the AI's source code.
This gives Alice the ability to
    spend \$\simcostExample{}k on hiring a programmer, to \textit{simulate} the AI in a sandbox and learn whether it is going to cooperate or defect.
    Crucially, we assume that the AI either does not learn whether it has been simulated or is unable to react to this fact.
We might hope that this will ensure that Alice and Bobble can reliably cooperate.
However,
    perhaps Alice will try to save on the \textit{simulation cost} and trust Bobble blindly instead
    --- and perhaps Bobble will bet on this scenario and instruct their AI to defect.

To analyze this modified game $\simTrustGame$,
    note that when Alice simulates,
        the only sensible followup is to trust Bobble if and only if the simulation reveals they instructed the AI to cooperate.
    As a result,
        the normal-form representation of $\simTrustGame$ is equivalent to the normal form of the original game $\trustGame$ with a single added action for Alice (Fig.\,\ref{fig:trust-game-nfg}).
Analyzing $\simTrustGame$ reveals that it has two types of Nash equilibria.
    In one,
        Bobble defects with high probability and
        Alice, expecting this, walks out without bothering to simulate.
    In the other,
        Bobble still sometimes defects ($\policy_\Bob(\defect) = \nicefrac{\simcostExample}{\defectLossAlice}$),
            but not enough to stop Alice from cooperating altogether.
        In response, Alice
            simulates often enough to stop Bobble from outright defection ($\policy_\Alice(\simulate) = 1 - \nicefrac{\coopUtilBob}{\defectUtilBob} = \nicefrac{5}{6}$),
            but also sometimes trusts Bobble blindly ($\policy_\Alice(\trust) = \nicefrac{\coopUtilBob}{\defectUtilBob} = \nicefrac{1}{6}$).
        In expectation, this makes Alice and Bobble better off by \$\eqUtilAlice{}k, resp. \$\eqUtilBob{}k relative to the (\textit{defect, walk-out}) equilibrium.

More generally,
we can also consider $\simTrustGame^\simcost$, a~parametrization of $\simTrustGame$
    where simulation costs some $\simcost \in \R$.
As shown in Figure~\ref{tab:TG-NE},
    the equilibria of $\simTrustGame^\simcost$ are similar to the special case $\simcost = \simcostExample{}$
    for a wide range of $\simcost$.

\medskip
\input{figures/TG_analysis.tex}

\paragraph{Generalizable properties of the trust game}
The analysis of Figure~\ref{fig:piecewise-linearity} illustrates several trends that hold more generally:
First, \textit{when simulation is subsidized,
    the simulation game turns into a ``pure commitment game'' where the simulated player is the Stackelberg leader} (Prop.\,\ref{prop:NE-for-extreme-simcost}\,(i)).
Conversely, \textit{when simulation is prohibitively costly,
    the simulation game is equivalent to the original game}
        (Prop.\,\ref{prop:NE-for-extreme-simcost}\,(ii)).
Third,
    the simulation game has a finite number of breakpoints between which
        individual equilibria change continuously
        --- more specifically, the simulator's strategy does not change at all while the simulated player's strategy changes linearly in $\simcost$
        (Prop.\,\ref{prop:piecewise_constant_linear}).
    Informally speaking, \textit{simulation games have piecewise constant/linear equilibrium trajectories}.
    A corollary of this observation is that \textit{it is \emph{not} the case that as simulation gets cheaper, the simulator must use it more and more often} (Fig.\,\ref{fig:piecewise-linearity}).
Fourth,
    the indifference principle implies that when the simulator simulates with a nontrivial probability (i.e., neither $0$ nor $1$),
        \textit{the value of information of simulating must be precisely equal to the simulation cost}.
    This also implies that
        \textit{any pure NE of the original game is also a NE of the simulation game} for any $\simcost \geq 0$
            (Prop.\,\ref{prop:VoI-and-specific-NE}).
Finally, we saw that
    \textit{at $\simcost = 0$,
        the \emph{outcome} of the simulation game becomes deterministic
        despite the \emph{strategy} of the simulator being stochastic}.
    (For example, in the NE where Bobble always cooperates, Alice always ends up trusting him --- either directly or after first simulating.)
    In Section~\ref{sec:complexity}, we show that this result holds quite generally but not always:
    By Theorem\,\ref{thm:NE-of-generic-NFG}, \textit{the equilibria of generic normal-form games with cheap simulation can be found in linear time}.

\paragraph{Different effects of simulation} 
There are games where simulation behaves similarly to the Trust Game above.
Indeed, in Theorem~\ref{thm:generealised-TG-no-tiebreaking},
    we prove that \textit{simulation leads to a strict Pareto improvement in generalized trust games} with generic payoffs (defined in Section~\ref{sec:theory:social-welfare}).
However, simulation can also affect games quite differently from what we saw so far.
For example, \textit{simulation can benefit either of the players at the cost of the other, or even be harmful to both of them}.
    Indeed, simulation
        benefits only the simulator in zero-sum games \edit{(Prop.\,\ref{prop:sim-in-zero-sum-games})},
        benefits only the simulated player in the Commitment Game (Fig.\,\ref{fig:commitment-game}),
        and harms both if cooperation is predicated upon the simulated player's ability to maintain privacy (Ex.\,\ref{ex:exp-exp}).
    In fact,
        there are even \textit{cases where the Pareto optimal outcome requires simulation to be neither free nor prohibitively expensive} (Ex.\,\ref{ex:nontrivial-simcost}).
Finally,
    with multiple, incompatible ways to cooperate,
    a \textit{game might admit multiple simulation equilibria}
        (i.e., multiple NE with $\policy_1(\simulate) > 0$; cf. Fig.\,\ref{fig:multiple-sim-eq}).

\begin{figure}[!tb]
    \centering
    \includegraphics[align=t,width=0.12\textwidth]{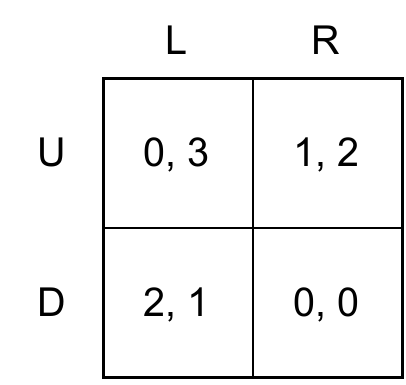}
    \ \ \ \ \ \ \ \ \ 
    \includegraphics[align=t,width=0.16\textwidth]{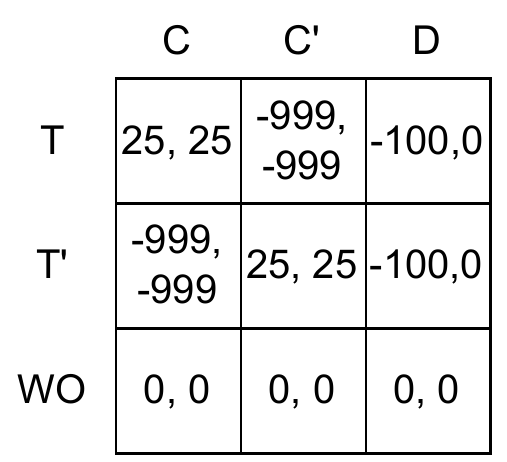}
    \caption{
        Left: Commitment game, where the row player prefers to \textit{not} be able to simulate.
        For details, see Example~\ref{ex:commitment-game-continued}.
        \\
        Right: A variant of Trust Game with multiple simulation NE.
    }
    \label{fig:commitment-game}\label{fig:multiple-sim-eq}
\end{figure}

%% file: figures/TG_analysis.tex
\begin{figure*}[tb]
\centering
\includegraphics[width=0.21\textwidth]{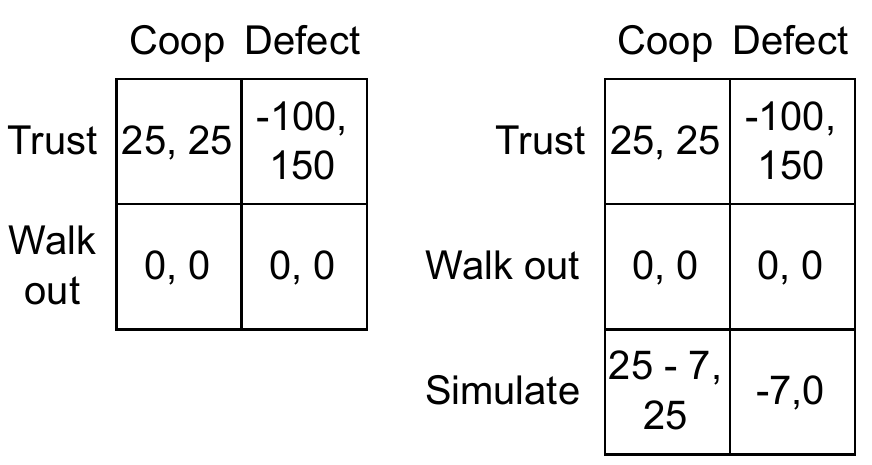}
$\ $
\includegraphics[width=0.75\textwidth]{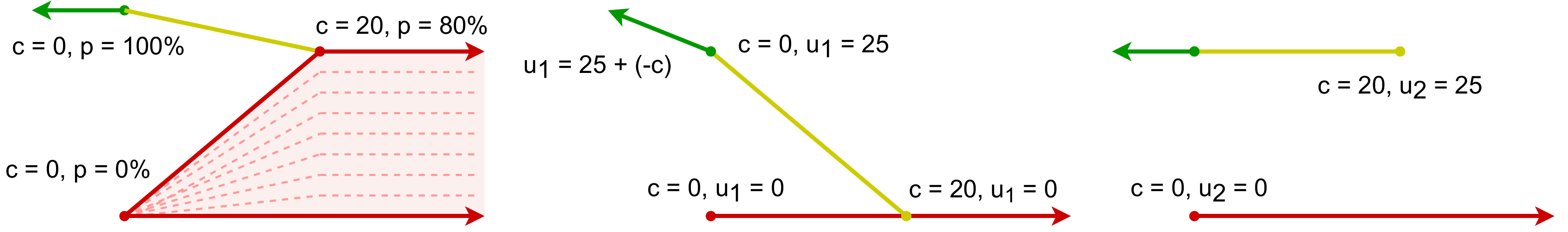}

\smallskip

\begin{tabular}{|c | l | l | l |}
\hline
 simulation cost
    & \multicolumn{3}{c |}{
        NE of $\simTrustGame^\simcost$, given as the probs. of
            (\textit{trust, walk-out, simulate}) and (\textit{cooperate, defect})
    }
    \\
 \hline
$(-\infty, 0)$
    & \ \ \ \ \ \ \ \ (\textit{simulate, cooperate}) &  & \\
\hline
$(0, \secondBreakpoint)$
    & $(\nicefrac{\coopUtilBob}{\defectUtilBob}, 0, 1 - \nicefrac{\coopUtilBob}{\defectUtilBob})$ and $(1 - \nicefrac{\simcost}{\defectLossAlice}, \nicefrac{\simcost}{\defectLossAlice})$
    & \textit{walk-out} and $(\nicefrac{\simcost}{\coopUtilAlice}, 1 - \nicefrac{\simcost}{\coopUtilAlice})$
    & (\textit{walk-out}, \textit{defect})
    \\
\hline
$(\secondBreakpoint, \infty)$
    & 
    & \textit{walk-out} and $(\nicefrac{\defectLossAlice}{\utilDiffAlice}, \nicefrac{\coopUtilAlice}{\utilDiffAlice})$ 
    & (\textit{walk-out}, \textit{defect})
    \\
    \hline
\end{tabular}
\caption{
    Top left:  The normal-form representation of the trust game from Figure~\ref{fig:trust-game-efg}, before and after adding simulation.
    \\
    Bottom: The extremal equilibria of $\simTrustGame^\simcost$.
        The non-extremal NE are precisely the convex combinations of the last two columns. 
    \\
    Top right:
        The cooperation probability and utilities under each of these NE.
        The non-extremal NE are light red, the dashed lines illustrate the NE trajectories from Proposition~\ref{prop:piecewise_constant_linear}.    
    Note that all the red NE (i.e., with $\policy_1(\walkOut) = 1$) yield $\utility_1 = \utility_2 = 0$.
    }
    \label{fig:piecewise-linearity}\label{tab:TG-NE}\label{fig:trust-game-nfg}
\end{figure*}

%% file: outline.tex
\subsection{Outline}


The remainder of the paper is structured as follows.
First, we
    recap the necessary background (Section~\ref{sec:background}).
In Section~\ref{sec:simgame-basics}, we formally define simulation games and describe their basic properties.
In Section~\ref{sec:structural}, we prove several structural properties of simulation games; while these are instrumental for the subsequent results, we also find them interesting in their own right.
Afterwards, we analyze
    the computational complexity of solving simulation games (Section~\ref{sec:complexity})
    and the effects of simulation on the players' welfare (Section~\ref{sec:theory:social-welfare}).
Finally, we
    review the most relevant existing work (Section~\ref{sec:related-work}),
    summarize our results,
    and discuss future work (Section~\ref{sec:discussion}).
The detailed proofs are presented in the appendix (which is only available in the arXiv version of the paper).

%% file: background.tex
A two-player \defword{normal-form game} (NFG)
    is a pair $\game = (\actions, \utility)$
    where
        $\actions = \actions_1 \times \actions_2 \neq \emptyset$ is a finite set of \defword{actions} 
        and $\utility = (\utility_1, \utility_2) : \actions \to \R^2$ is the \defword{utility function}.
We use \Plone{} and \Pltwo{} as shorthands for ``player one'' and ``player two''.
For finite $X$, $\Delta(X)$ denotes the set of all probability distributions over $X$.
A \defword{strategy} (or policy) \defword{profile} is a pair $\policy = (\policy_1, \policy_2)$ of \defword{strategies} $\policy_\pl \in \Delta(\actions_\pl)$.
We denote the set of all strategies as $\policies = \policies_1 \times \policies_2$.
$\policy_\pl$ is \defword{pure} if it has support $\supp (\policy_\pl)$ of size $1$. We identify such strategy with the corresponding action.

For $\policy \in \policies$,
    $\utility_\pl(\policy)
        := \sum_{(\action, \actionOpp) \in \actions}
            \policy_1(\action)\policy_2(\actionOpp)
            \utility_\pl(\action, \actionOpp)
    $
    is the \defword{expected utility} of $\policy$.
$\policy_1$ is said to be a \defword{best response} to $\policy_2$ if
    $\policy_1 \in \arg \max_{\policy'_1 \in \policies_1} \utility_1(\policy'_1, \policy_2)$;
$\br(\policy_2)$
    denotes the set of all pure best responses to $\policy_2$.
Since the \defword{best-response utility} ``$\utility_1 (\br, \symbolPlaceholder )$'' is uniquely determined by $\policy_2$, we denote it as
    $
        \utility_1 (\br, \policy_2)
        := \max_{\action \in \actions_1} \utility_1(\action, \policy_2)
    $.
(The analogous definitions apply for \Pltwo{} and $\policy_1$.)
A \defword{Nash equilibrium} (NE)
    is a strategy profile $(\policy_1, \policy_2)$ under which each player's strategy is a best response to the strategy of the other player.
    We use $\NE(\game)$ to denote the set of all Nash equilibria of $\game$.

\medskip

A \defword{pure-commitment equilibrium} (cf. \cite{von1934marktform}) is, informally,
        a subgame-perfect
    equilibrium of the game in which
    the leader first commits to a pure action,
    after which the follower sees the commitment and best-responds, possibly stochastically.
Since our formalism will assume that \Plone{} is the simulator, we naturally encounter situations where \Pltwo{} acts as the leader.
Formally, we will use $\pureSE(\game)$ to denote
    all pairs $(\brSelector, \actionOpp)$ where
        the \defword{optimal commitment} $\actionOpp \in \actions_2$ and
        \Plone{}'s best-response policy $\brSelector : \actionOppAlt \in \actions_2 \mapsto \brSelector(\actionOppAlt) \in \Delta(\br(\actionOppAlt)) \subseteq \Delta(\actions_1)$
        satisfy
            $\actionOpp \in \arg \max_{\actionOppAlt \in \actions_2} \, \E_{\action \sim \brSelector(\actionOppAlt)} \, \utility_2(\action, \actionOppAlt)$.

\medskip

\edit{
Below, we sometimes restrict the analysis to particular classes of NFGs.
To motivate the first one,
    recall that a property is said to be generic (typical)
        if it holds for almost all elements of a set \cite[1.35]{rudinRealAndComplexAnalysis}

\begin{definition}[Generic games]\label{not:generic-games}
We say that a statement $P$ holds \defword{for games with generic payoffs} if, among games whose payoffs are sampled i.i.d. from the uniform distribution over $[0, 1]$, \textit{P} holds with probability 1.
\end{definition}}

\edit{
Since different joint actions in a generic-payoff NFG necessarily yield different payoffs,
    these games are a special case of the following more general class:

\begin{definition}[No best-response utility tiebreaking]\label{def:no-br-tiebreaking}
An NFG $\game$ is said to admit \defword{no best-response utility tiebreaking} by \Plone{} if
    for every pure strategy $\actionOpp$ of \Pltwo{}, any two pure best-responses $\action, \actionAlt \in \br(\actionOpp)$ give \Pltwo{} the same utility, i.e. $\utility_2(\action, \actionOpp) = \utility_2(\actionAlt, \actionOpp) =: \utility_2(\br, \actionOpp)$.
\end{definition}
}

\noindent
Note that if $\game$ satisfies Def.~\ref{def:no-br-tiebreaking},
    any pure-commitment equilibrium $\policy \in \pureSE(\game)$ can be identified with a joint action $(\action, \actionOpp)$ s.t.
    $\action \in \br(\actionOpp)$ and
    $\actionOpp \in \arg \max_{\actionOpp \in \actions_2} \utility_2(\br, \actionOpp)$.



%% file: theory_basic_properties.tex
In this section, we formally define simulation games and describe their basic properties.
To streamline this initial investigation of simulation games, we assume that
    when the simulator learns the other agent's action, they
        always best-respond to it
    --- that is, they will not execute non-credible threats \cite{MAS}.
    \edit{(However, this assumption somewhat limits the applicability of the results, and we consider moving beyond it a worthwhile future direction.)}

\edit{
\begin{notation}
For a two-player NFG $\game$,
    $
        \brSelectors
        := \{
            \brSelector : \actions_2 \to \Delta(\actions_1) \mid \forall \actionOpp \in \actions_2 : \brSelector(\actionOpp) \in \Delta(\br(\actionOpp))
        \}
    $,
    resp. $\brSelectorsPure \subset \brSelectors$,
    is the set of all stochastic, resp. pure \defword{best-response policies}.
\end{notation}
}

\begin{definition}[Simulation game]\label{def:simgame}
(1) For a \defword{simulation cost} $\simcost \in \R$,
    the \defword{simulation game} $\simgame^{\simcost,\all}$ is defined as
        the NFG that is identical to $\game$,
        \edit{
        except that \Plone{} additionally has access to ``simulation'' actions $\simulate_{\brSelector}$, $\brSelector \in \brSelectorsPure$,
            s.t.
            $\utility_1(\simulate_{\brSelector}, \actionOpp) := \utility_1(\br, \actionOpp) - \simcost$,
            $\utility_2(\simulate_{\brSelector}, \actionOpp) := \utility_2(\brSelector(\actionOpp), \actionOpp)$.}

(2) For a fixed $\brSelector \in \brSelectors$,
    $\simgame^\simcost := \simgame^{\simcost,\brSelector}$ denotes the game where
    \Plone{} has a single additional action $\simulate := \simulate_{\brSelector}$
    with
        $\utility_1(\simulate, \actionOpp) := \utility_1(\br, \actionOpp) - \simcost$,
        $\utility_2(\simulate, \actionOpp) := \E_{\action \sim \brSelector(\actionOpp)} \utility_2(\action, \actionOpp)$.
\end{definition}

\noindent
We refer to \Plone{} as the \defword{simulator} and to \Pltwo{} as the \defword{simulated player}.
When the exact value of $\simcost$ is unspecified or unimportant, we write $\simgame$ instead of $\simgame^\simcost$.

\subsection{Basic Properties}\label{sec:sub:basic-properties}

\edit{
In the remainder of this paper, we will only study simulation games in the context of a fixed best-response policy.
To justify this decision, note that the variants (1) and (2) of \Cref{def:simgame} are equivalent for most games:

\begin{restatable}{lemma}{simgameDefEquivalence}\label{lem:simgame-def-equivalence}
If (and only if) $\game$ admits no best-response utility tiebreaking by \Plone{},
$\simgame^{\simcost,\all}$ and $\simgame^\simcost$ are identical up to the existence of duplicate actions.
\end{restatable}

Moreover, the problem of solving general simulation games can be reduced to the problem of solving simulation games for a fixed best-response policy:

\begin{restatable}{lemma}{reductionToFixedBRpolicy}\label{lem:reduction-to-fixed-br-policy}
$\NE(\simgame^{\simcost,\all}) \setminus \NE(\game)$ (i.e., the new NE introduced by adding simulation)
can be written as a disjoint union
$
    \dot{\bigcup}_{\brSelector \in \brSelectors}
    \NE(\simgame^{\simcost, \brSelector})
    \setminus
    \NE(\game)
$.
\end{restatable}
}

The next observation we make (Proposition~\ref{prop:NE-for-extreme-simcost}) is that
    if simulation is too costly, then it is never used and the simulation game $\simgame$ becomes strategically equivalent to the original game $\game$.
    Conversely, if simulation is subsidized (i.e., a negative simulation cost), then \Plone{} will always use it,
        which effectively turns $\simgame$ into a pure-commitment game with \Pltwo{} moving first.
    (The situation is similar when simulation is free but not subsidized,
        except that this allows for additional equilibria where the simulation probability is less than $1$.)

\begin{restatable}[Equilibria for extreme simulation costs]{proposition}{basicNEresults}\label{prop:NE-for-extreme-simcost}
In any simulation game $\simgame$, we have:
\begin{enumerate}[label={(\roman*)}]
    \item For $\simcost < 0$, simulating is a strongly dominant action.
    \\
        In particular, $\NE(\simgame^\simcost) \subseteq \pureSE(\game)$.\footnotemark{}
    
    \item For 
        $
            \simcost
            >
            \underset{\action \in \actions_1,  \actionOpp \in \actions_2}{\max}
                \utility_1(\action, \actionOpp)
            - 
            \underset{\policy_1 \in \policies_1}{\max}
            \underset{\policy_2 \in \policies_2}{\min}
                \utility_1(\policy_1, \policy_2)
        $,\\
        $\simulate$ is a strictly dominated action.
        
        In particular, $\NE(\simgame^\simcost) = \NE(\simgame)$.
\end{enumerate}
\end{restatable}

\footnotetext{
    If we allowed \Plone{} to consider all possible best-response policies,
        $\NE(\simgame^\simcost) \subseteq \pureSE(\game)$ would turn into equality.
}

\subsection{Information-Value of Simulation}\label{sec:sub:VoI}

The following definition measures the extra utility that the simulator can gain by using the knowledge of the other player's strategy:

\begin{definition}[Value of information of simulation]\label{def:VoI}
The \defword{value of information of simulation} for $\policy_2 \in \policies_2$ is
\begin{align*}
    \VoI(\policy_2)
        :=
        &
        \left( \sum\nolimits_{\actionOpp \in \actions_2}
            \!\!\!
            \policy_2(\actionOpp)
            \max_{\action \in \actions_1}
                \utility_1(\action, \actionOpp) \right)
        \!
        -
        \!\!\!\,
        \max_{\policy_1 \in \policies_1}
            \!
            \utility_1(\policy_1, \policy_2)
        .
\end{align*}
\end{definition}

\begin{restatable}{lemma}{VoIElementary}\label{lem:VoIandSim}
$
    \forall \policy_2
    :
    \utility_1(\simulate, \policy_2)
    = \utility_1(\br, \policy_2) + \VoI(\policy_2) - \simcost
$.
\end{restatable}

\noindent
\Cref{lem:VoIandSim} implies
    that $\VoI(\policy_2)$ always lies between $0$ and the difference between maximum possible $\utility_1$ and \Plone{}'s maxmin value.
Moreover, to make \Plone{} simulate with a non-trivial probability,
    \Pltwo{} needs to pick a strategy whose value of information is equal to the simulation cost:

\begin{restatable}[$\VoI$ is equal to simulation cost]{lemma}{VOIequalsSIMCOST}
    \label{lem:VoI}
(1) For any $\policy \in \NE(\simgame^\simcost)$, we have
    $\policy_1(\simulate) \in (0, 1) \implies \VoI(\policy_2) = \simcost$.
(2) Moreover,
    unless $\game$ admits multiple optimal commitments of \Pltwo{} that do not have a common best-response,
    any $\policy \in \NE(\simgame^0)$ has $\VoI(\policy_2) = 0$.
\end{restatable}

\noindent
(Where, in (2), a set of actions having a common best-response means that $\bigcap_{\actionOpp \in \actionSetOpp} \br(\actionOpp) \neq \emptyset$.)

%% file: theory_structural.tex
We now review several structural properties that appear in simulation games because of the special nature of the simulation action.These results will prove instrumental when determining the complexity of simulation games (Sec.\,\ref{sec:complexity}) and predicting the impact of simulation on the players' welfare (Sec.\,\ref{sec:theory:social-welfare}).
Moreover, we find these results interesting 
in their own right.

%
%
The first of these properties is that a change of the simulation cost \textit{typically} results in a very particular change in a Nash equilibrium of the corresponding game:
The strategy of the simulating player (\Plone) doesn't change at all, while the simulated player's strategy changes linearly.
However, to be technically accurate, we need to make two disclaimers.
    First, there is a finite number of ``atypical'' values of $\simcost$, called breakpoints,
        where the nature of the NE strategies changes discontinuously.\footnotemark{}
    Second, there can be multiple equilibria, which complicates the formal description of the result.
\footnotetext{
    \label{fn:breakpoint-equilibria}
        While all of the non-breakpoint equilibria extend to the corresponding breakpoints as limits (Definition~\ref{def:limit-eq}),
        the breakpoints might also admit additional non-limit equilibria, typically convex combinations of the limits
            (cf. \Cref{fig:piecewise-linearity}).
}

\begin{restatable}[Trajectories of simulation NE are piecewise constant/linear]{proposition}{piecewiseLinear}\label{prop:piecewise_constant_linear}
For every $\game$,
    there is a finite set of simulation-cost breakpoint values
        $- \infty = \exceptionPoint_{-1} < 0 = \exceptionPoint_0 < \exceptionPoint_1 < \dots < \exceptionPoint_k < \exceptionPoint_{k+1} = \infty $
    such that the following holds:
For every $\simcost_0 \in ( \exceptionPoint_l, \exceptionPoint_{l+1} )$
    and every $\policy^{\simcost_0} \in \NE(\simgame^{\simcost_0})$,
there is a linear mapping $\NEtrajectory_2 : \simcost \in [ \exceptionPoint_l, \exceptionPoint_{l+1} ] \mapsto \policy^\simcost_2 \in \policies_2$
such that
    $\NEtrajectory_2(\simcost_0) = \policy_2^{\simcost_0}$ and
    $(\policy^{\simcost_0}_1, \NEtrajectory_2(\simcost)) \in \NE(\simgame^\simcost)$ for every $\simcost \in [\exceptionPoint_l, \exceptionPoint_{l+1} ]$.
\end{restatable}

Since we were not able to find any existing result that would immediately imply this proposition, we provide our own proof in the appendix.
    However, a related result in the context of parameterized linear programming appears in \cite[Prop.\,2.3]{parametrizedLPs}.
As an intuition for why this result holds,
    recall that in an equilibrium, each player uses a strategy that makes the other player indifferent between the actions in their support.
    Since \Pltwo's payoffs are not affected by $\simcost$, \Plone{} should keep their strategy constant to keep \Pltwo{} indifferent, even when $\simcost$ changes.
    Similarly, increasing $\simcost$ linearly decreases \Plone's payoff for the simulate action, so \Pltwo{} needs to linearly adjust their strategy to bring \Plone's payoffs back into equilibrium.

A particular corollary of Proposition~\ref{prop:piecewise_constant_linear} is that
    while one might perhaps expect simulation will gradually get used more and more as it gets more affordable,
    this is in fact not what happens
    --- instead, the simulation rate is dictated by the need to balance the unchanging tradeoffs of the other player.

\medskip

The second structural property of simulation games is
    the following refinement of Proposition~\ref{prop:NE-for-extreme-simcost}:

\begin{restatable}[Gradually recovering the NE of $\game$]{proposition}{VoIandSpecNE}
    \label{prop:VoI-and-specific-NE}
Let $\policy$ be a NE of $\game$.
Then $\policy$,
    as a strategy in $\simgame^\simcost$ with $\policy_1(\simulate) := 0$,
    is a NE
    precisely when $\simcost \geq \VoI(\policy_2)$.

In particular, $\VoI(\policy_2)$ is a breakpoint of $\game$.
\end{restatable}

\noindent
Together, these two results imply that
    with $\simcost = 0$,
        $\simgame^0$ may have no NE in common with $\game$.
    As we increase $\simcost$,
        the NE of $\game$ gradually appear in $\simgame^\simcost$ as well,
        while the simulation equilibria of $\simgame$
            (i.e., those with $\policy_1(\simulate) > 0$)
        gradually disappear,
    until eventually $\NE(\simgame^\simcost) = \NE(\game)$.

\subsection{Equilibria for Cheap Simulation}

By combining the concept of value of information with the piecewise constancy/linearity of simulation equilibria,
    we are now in a position to give a more detailed description of Nash equilibria of games where simulation is cheap.
First, we identify the equilibria of $\simgame$ with $\simcost = 0$ that might be connected to the equilibria for $\simcost > 0$:

\begin{definition}[Limit equilibrium of $\simgame$]\label{def:limit-eq}
A policy profile $\policy^0$ is a \defword{limit equilibrium} (at $\simcost = 0$) of $\simgame$
    if it is a limit of some $\policy^{\simcost_n} \in \NE(\simgame^{\simcost_n})$ where $\simcost_n \to 0_+$.
\end{definition}

\noindent
As witnessed by the Trust Game (and Table~\ref{tab:TG-NE} in particular), not every NE of $\simgame^0$ is a limit equilibrium.
Note that this definition
    automatically implies a stronger condition:

\begin{restatable}{lemma}{limitEquilibriaWLOG}\label{lem:limit-eq-alternative-def}
For any limit equilibrium $\policy^0$ of $\simgame$,
there is some $\exceptionPoint > 0$ and $\policy^\exceptionPoint_2$ such that
    for every $\simcost \in [0, \exceptionPoint]$,
    $(\policy^0_1, (1 - \frac{\simcost}{\exceptionPoint}) \policy^0_2
            + \frac{\simcost}{\exceptionPoint} \policy^\exceptionPoint_2)$
    is a NE of $\simgame^\simcost$.
\end{restatable}

The following result shows that cheap-simulation equilibria have a very particular structure.
    Informally, every such NE corresponds to a ``baseline'' limit equilibrium $\policy^\baseline$ and \Pltwo{}'s ``deviation policy'' $\policy^\deviate_2$.
    As the simulation cost increases, \Pltwo{} gradually deviates away from their baseline, which forces \Plone{} to randomize between their baseline and simulating.
    While the technical formulation can seem daunting, all of the conditions in fact have quite intuitive interpretations that can be used for locating the simulation equilibria of small games by hand.

\newcommand{\likesBaselineLabel}{(D_2^>)}
\newcommand{\indiffLabel}{(D_2^=)}
\newcommand{\likesSimLabel}{(D_2^<)}

\begin{restatable}[Structure of cheap-simulation equilibria]{lemma}{limitEquilibriaVoI}
    \label{lem:limit-NE-structure}
Let $\simcost_0 \in (0, \exceptionPoint_1)$
and suppose that 
    $\game$ admits no best-response utility tiebreaking by \Plone{}.
%
Then any $\policy \in \NE(\simgame^{\simcost_0})$ with $\policy_1(\simulate) \in (0, 1)$
    is of the form $\policy = (\policy_1, \policy^{\simcost_0}_2)$,
    where
    \begin{align*}
        \policy_1 & \ = \ (1 - \policy_1(\simulate) ) \cdot \policy^\baseline_1 + \policy_1(\simulate) \cdot \simulate \\
        \policy^\simcost_2 & \ = \ \ \ (1 - \slope \simcost) \ \, \cdot \policy^\baseline_2 + \ \, \slope \simcost \ \,  \cdot \policy^\deviate_2
        , \ \ \ 
        \alpha > 0
        ,
    \end{align*}
and the following holds:
\begin{enumerate}[label={(\roman*)}]
    \item For every $\simcost \in [0, \exceptionPoint_1]$, $(\policy_1, \policy^\simcost_2) \in \NE(\simgame^\simcost)$.
    \item $\policy^\baseline \in \policies$ is some \defword{baseline policy} that satisfies:
        \begin{itemize}[leftmargin=*]
            \item[(B1)] every action in the support of $\policy^\baseline_1$ is a best-response to every action from $\supp(\policy^\baseline_2)$;
            \item[(B2)] every action in the support of $\policy^\baseline_2$ is an optimal commitment by \Pltwo{} conditional on \Pltwo{} only using strategies that satisfy (B1).
        \end{itemize}
    \item $\policy^\deviate_2 \in \policies_2$ is some \defword{deviation policy} that satisfies:
        \begin{itemize}[leftmargin=*]
        \item[(D1)] No $\action \! \in \! \supp(\policy^\baseline_1)$ lies in $\br(\deviation)$ for all $\deviation \! \in \! \supp(\policy^\deviate_2)$.
        \item[(D2)] Every $\deviation \in \supp(\policy^\deviate_2)$ satisfies one of
            \begin{align*}
                \utility_2(\policy^\baseline_1, \deviation) > \utility_2(\policy^\baseline) > \utility_2(\br, \deviation)
                    \phantom{.}
                    \hspace{3em} \likesBaselineLabel
                    \\
                \utility_2(\policy^\baseline_1, \deviation) = \utility_2(\policy^\baseline) = \utility_2(\br, \deviation)
                    \phantom{.}
                    \hspace{3em} \indiffLabel
                    \\
                \utility_2(\policy^\baseline_1, \deviation) < \utility_2(\policy^\baseline) < \utility_2(\br, \deviation)
                .
                    \hspace{3em} \likesSimLabel
            \end{align*}           
        \item[(D3)] If $\deviation \in \supp(\policy^\deviate_2)$ satisfies $\likesBaselineLabel$, resp. $\likesSimLabel$,\\
            it maximizes the \emph{attractiveness ratio} $\deviationRatio_\deviation$, resp. $\deviationRatioInverse_\deviation$
            \begin{align*}
                \frac{
                    \utility_2(\policy^\baseline_1, \deviation') - \utility_2(\policy^\baseline)
                }{
                    \utility_2(\policy^\baseline) - \utility_2(\br, \deviation')
                }
                \text{ resp. }
                \frac{
                    \utility_2(\br, \deviation') - \utility_2(\policy^\baseline)
                }{
                    \utility_2(\policy^\baseline) - \utility_2(\policy^\baseline_1, \deviation')
                }
            \end{align*}
            among all $\deviation' \in \actions_2$ that satisfy $\likesBaselineLabel$, resp. $\likesSimLabel$.
        \end{itemize}
\end{enumerate}
\end{restatable}

\noindent
In a generic game, these conditions even imply that both the baseline and deviation policies are pure.

\begin{restatable}[Equilibria with binary supports]{theorem}{supportSizeTwo}\label{thm:generic-games-support-size}
Let $\game$ be a game with generic payoffs
    and $\simcost \in (0, \exceptionPoint_1)$.
    Then all NE of $\simgame^\simcost$ are either pure or have supports of size two.
\end{restatable}

%% file: theory_computational.tex
We now investigate the difficulty of solving simulation games.
Since many of the results hold for multiple solution concepts, we formulate them using the phrase ``solving a game'',
with the understanding that this refers to either finding all Nash equilibria, or a single NE, or a single NE with a specific property (e.g., one with the highest social welfare).
For a specific game $\game$, we will also use
    $ -\infty < 0 < \exceptionPoint_1 < \dots < \exceptionPoint_{k} < \infty$
    to denote the breakpoints of $\simgame$ (given by Proposition\,\ref{prop:piecewise_constant_linear}).
\journal{
First, note that the number of $\exceptionPoint_i$-s could be exponential:

\newcommand{\harmonicMean}{H}
\newcommand{\indexSubset}{I}
\newcommand{\payoffA}{x}
\newcommand{\payoffB}{y}
\newcommand{\payoffsA}{\vec \payoffA}
\newcommand{\payoffsB}{\vec \payoffB}
\newcommand{\payoffSubsetA}{\payoffsA_\indexSubset}
\newcommand{\payoffSubsetB}{\payoffsB_\indexSubset}
\newcommand{\prodA}{\mathbf{\payoffA}}
\newcommand{\prodB}{\mathbf{\payoffB}}

\begin{restatable}[Games can have exponentially many breakpoints]{proposition}{exponentialBreakpoints}
    \label{prop:breakpoints-exponential}
For every $n \in \N$, there is a two-player NFG $\game$ with $| \actions_1 | = | \actions_2 | = n$ such that $\simgame$ has at least $2^n$ breakpoints.
\end{restatable}

\begin{figure}[!tb]
    \centering
    \includegraphics[width=0.24\textwidth]{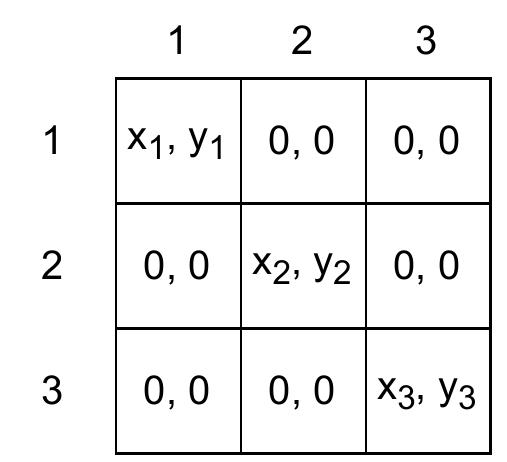}
    \caption{
        Cafés in Paris: Alice and Bob want to meet but they need to coordinate on which café to go to.
        We assume that $\payoffA_i, \payoffB_i > 0$ for every $i$.
        The actual game has $n \in \N$ actions.
    }
    \label{fig:cafes-in-paris}
\end{figure}
}

\edit{As an upper bound on the complexity of solving simulation games, their definition immediately yields that:}

\begin{restatable}[Simulation games are no harder than general games]{proposition}{GsimEquallyHard}\label{prop:NE-complexity-basic}
Solving $\simgame^\simcost$ is at most as difficult as solving a normal-form game where \Plone{} has one more action than in $\game$.
\end{restatable}

\edit{For extreme values of $\simcost$, Prop.~\ref{prop:NE-for-extreme-simcost} implies the following:}

\begin{restatable}[Solving $\simgame$ for extreme $\simcost$]{proposition}{complexityForExtremeSimcost}\label{prop:NE-complexity-extreme}
(i) For $\simcost < 0$, the time complexity of solving $\simgame^\simcost$ is $\bigO(|\actions|)$.

\noindent
(ii) For $\simcost > \exceptionPoint_{k}$, the time-complexity of solving $\simgame^\simcost$ is the same as the time-complexity of solving $\game$.
\end{restatable}

In contrast with Proposition~\ref{prop:NE-complexity-extreme}\,(ii),
    finding the equilibria at low simulation costs is straightforward
    \edit{if we restrict our attention to generic-payoff NFGs}:

\begin{restatable}[Cheap-simulation equilibria in generic games]{theorem}{genericGames}\label{thm:NE-of-generic-NFG}
Let $\game$ be a NFG with generic payoffs
and $\simcost \in (0, \exceptionPoint_1)$.
Then the time complexity of finding all equilibria of $\simgame^\simcost$ is $\bigO(|\actions|)$.
\end{restatable}



\journal{
However, without generic payoffs, simulation games can be as hard to solve as fully general games.

\newcommand{\deviationProb}{\alpha \simcost}
\newcommand{\gameModif}{\widetilde \game}
\newcommand{\actionsModif}{\widetilde \actions}
\newcommand{\simgameModif}{\widetilde \game_\simSubscript}

\begin{restatable}[Solving general simulation games is hard]{theorem}{simgameHardness}\label{thm:simgame-hardness-general}
For an NFG $\game$ whose utilities lie in $[0, 1]$,
    let $\gameModif$ be as in \Cref{fig:trust-game-hardness}.
Then for $\simcost \in (0, \nicefrac{1}{3})$,
    $\NE(\simgameModif^\simcost)$ can be identified with $\NE(\game)$.
Moreover,
    whenever $\simgameModif^\simcost$ has some NE $\policy$ with
        $\policy_1(\simulate) = \simProb$,
        $\policy_2(\deviate) = \slope \simcost$,
    $\game$ must have some NE strategy $\policy^\game$ with
        $\utility_1(\policy^\game) = 2 - \nicefrac{1}{\slope}$,
        $\utility_2(\policy^\game) = \frac{1}{1 - \simProb} - 2$.
\end{restatable}

\begin{figure}[!tb]
    \centering
    \includegraphics[width=0.24\textwidth]{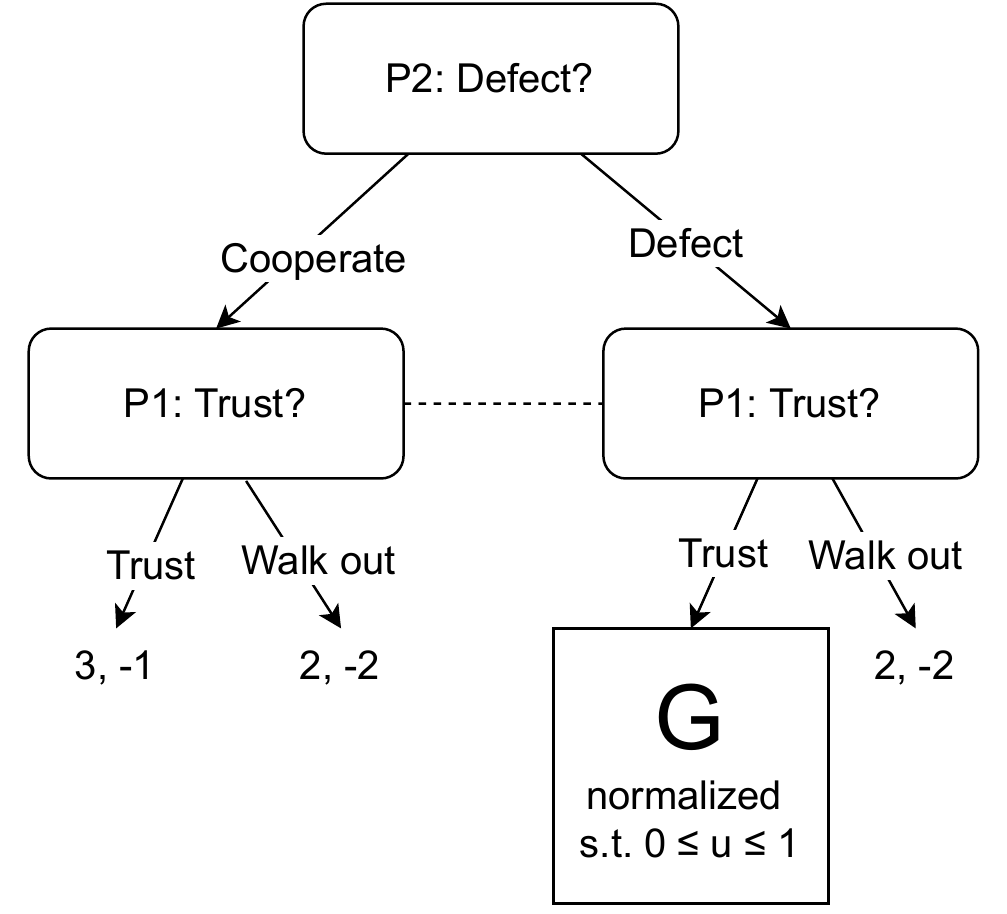}
    \caption{
        The game $\gameModif$ which shows that for any $\game$, there is a similarly-sized simulation game that is as hard to solve as $\game$.
    }
    \label{fig:trust-game-hardness}
\end{figure}
}




\edit{
Finally, it is also generally difficult to determine whether simulation is beneficial or not:

\begin{restatable}{theorem}{usfulnessIsNPhard}\label{thm:determining-usefulness-is-hard}
For a general NFG $\game$ and $\simcost \in \R$, it is co-NP-hard to determine whether
    there is $\policy \in \NE(\simgame^\simcost) \setminus \NE(\game)$ s.t.
    $
        \forall \rho \in \NE(\game) :
        \utility_1(\policy) \geq \utility_1(\rho)
        \And
        \utility_2(\policy) \geq \utility_2(\rho)
    $.
\end{restatable}
}

%% file: theory_social_welfare.tex
\edit{
As we saw in \Cref{thm:determining-usefulness-is-hard},
    there is no simple method for determining whether introducing simulation into a general game will be socially beneficial.
However, this does not rule out the possibility of identifying particular sub-classes of games where simulation is useful or harmful.
}
We now first
    confirm the hypothesis that
    simulation is beneficial in settings where the only obstacle to cooperation is the missing trust in the simulated player.
We then \edit{give specific examples to illustrate that} in general games,
    simulation can also benefit either player at the cost of the other,
    or even be harmful to both.


\subsection{Simulation in Generalized Trust Games}

We now show that
    when the \textit{only} obstacle to cooperation is the lack of trust in the possibly-simulated player,
    simulation enables equilibria that improve the outcome for both players.

\begin{definition}[Generalized trust games]\label{def:generalised-TG}
A game $\game$ is said to be a \defword{generalized trust game}
    if any pure-commitment equilibrium (where \Pltwo{} is the leader) is a strict Pareto improvement over any $\policy \in \NE(\game)$.
\end{definition}

\edit{
\begin{restatable}[Simulation in trust games helps]{theorem}{generalisedTGnoTiebreaking}\label{thm:generealised-TG-no-tiebreaking}
Let $\game$ be a generalized trust game
    that admits no best-response utility tiebreaking by \Plone{}.
Then for all sufficiently low $\simcost$,
    $\simgame^\simcost$ admits a Nash equilibrium with $\policy_1(\simulate) > 0$
    that is a strict Pareto improvement over any NE of $\game$.
\end{restatable}
}


\begin{proof}[Proof sketch]
We construct a NE where 
    \Pltwo{} mixes between their optimal commitment $\actionOpp$ (from the pure-commitment equilibrium corresponding to $\game$) and some deviation $\deviation$
    while \Plone{} mixes between their best-response to $\actionOpp$ and simulating.
We show that $(\action, \actionOpp)$ forms the baseline policy of this simulation equilibrium, which implies that as $\simcost \to 0_+$, this NE eventually strictly Pareto-improves any NE of $\game$.
(And the fact that $(\action, \actionOpp)$ cannot be a NE of $\game$ ensures the existence of a suitable $\deviation$.)
\end{proof}

\subsection{Simulation in General Games}

We now investigate the relationship between simulation cost and the players' payoffs in \textit{general} games.
We start by listing the two general trends that we are aware of.

The first of the general results is that for the extreme values of $\simcost$, the situation is always predictable:
For $\simcost < 0$,
    \Plone{} always simulates (Prop.\,\ref{prop:NE-for-extreme-simcost})
    and making simulation cheaper will increase their utility without otherwise affecting the outcome.
Similarly, when $\simcost$ is already so high that \Plone{} never simulates,
    any further increase of $\simcost$ makes no additional difference.

Second, if \Pltwo{} could choose the value of $\simcost$, they would generally be indifferent between all the values within a specific interval $(\exceptionPoint_i, \exceptionPoint_{i+1})$.
    Indeed, this follows from Proposition~\ref{prop:piecewise_constant_linear},
        which implies that \Pltwo{}'s utility remains constant between any two breakpoints of $\simgame$.

\medskip

The Examples~\ref{ex:cheap-cheap}-\ref{ex:exp-exp} illustrate that
    the players might both agree and disagree about their preferred value of $\simcost$,
    and this value might be both low and high.

\begin{example}[Both players prefer cheap simulation]\label{ex:cheap-cheap}
    In the Alice and Bobble game from Figure~\ref{fig:trust-game-nfg}, each player's favoured NE exists for $\simcost = 0$.
\end{example}

\begin{example}[Only simulator prefers cheap simulation]\label{ex:cheap-exp}
    Consider the ``unfair guess-the-number game'' where
        each player picks an integer between $1$ and $\guessTheNumberConstant$.
        If the numbers match, \Pltwo{} pays $1$ to \Plone{}.
        Otherwise, \Plone{} pays $1$ to \Pltwo{}.
    In this game, \Pltwo{} clearly prefers simulation to be prohibitively costly while \Plone{} prefers as low $\simcost$ as possible.
\end{example}

\edit{
In fact, \Cref{ex:cheap-exp} extends to all zero-sum games:

\begin{restatable}{proposition}{simInZeroSumGames}\label{prop:sim-in-zero-sum-games}
If a zero-sum $\game$ has NE utilities $(v, -v)$,
    then
    $\forall \simcost$ $\forall \policy \in \NE(\simgame^\simcost)$:
    $\utility_1(\policy) \geq v$, $\utility_2(\policy) \leq -v$.    
\end{restatable}
}

\begin{example}[Only simulator prefers expensive simulation]\label{ex:commitment-game-continued}\label{ex:exp-cheap}
    In the commitment game (Figure~\ref{fig:commitment-game}),
        introducing free simulation
        creates a second NE in which \Plone{} is strictly worse off
        and stops the original NE from being trembling-hand perfect.
    If simulation were subsidized, the original simulator-preferred NE would disappear completely.
    (In fact, with $\simcost > 0$ that is not prohibitively costly,
        the situation is similar to the $\simcost = 0$ case.)
    In summary, this shows that
        simulation can hurt the simulator, even when using it is free (or even subsidized) and voluntary.
\end{example}

\begin{example}[Both players prefer expensive simulation]\label{ex:exp-exp}
    Consider a Joint Project game where
        \Plone{} proposes that \Pltwo{} collaborates with them on a startup.
    If \Pltwo{} accepts, their business will be successful,
        yielding utilities $\utility_1 = \utility_2 = 100$.
        \Pltwo{} then picks a secure password ($pw \in \{1, \dots, 26 \}^{4}$) and puts their profit in a savings account protected by that password.
    Finally, \Plone{} can either do nothing or try to guess \Pltwo{}'s password ($g \in \{1, \dots, 26\}^{4}$) and steal their money.
        Successfully guessing the password would result in utilities $\utility_1 = 200$, $\utility_2 = - 10$,
            where the $-10$ comes from opportunity costs.
        However, if \Plone{} guesses wrong, they will be caught and sent to jail, yielding utilities $\utility_1 = -999$, $\utility_2 = 123$ \cite{smith2009exploring}.
    
    Without simulation, the NE of this game is for the players to collaborate and for \Plone{} to not attempt to guess the password.
    However, with cheap enough simulation, \Plone{} would simulate \Pltwo{}'s choice of password and steal their money --- and \Pltwo{}, expecting this, would not agree to the collaboration in the first place.
    As a result, both players would prefer simulation to be prohibitively expensive.
\end{example}

\begin{example}[The preferences depend on equilibrium selection]\label{ex:eq-selection}
    Consider various mixed-motive games
        such as the Threat Game (e.g., \cite[Sec.\,3-4]{clifton2020cooperation}),
        Battle of the Sexes, or Chicken (e.g., \cite{MAS}).
    Generally, these games have
        one pure NE that favours \Plone{},
        a second pure NE that favours \Pltwo{},
        and a mixed NE that is strictly worse than either of the pure equilibria for both \Plone{} and \Pltwo{}.
    By introducing subsidized simulation into such a game,
        we eliminate both the simulator-favoured pure NE and the dispreferred mixed NE.
    This can be bad, neutral, or even good news for the simulator,
        depending on which of the NE would have been selected in the original game.
    Somewhat relatedly, introducing subsidized simulation destroys the suboptimal equilibria in Stag Hunt and Coordination Game (e.g., \cite{MAS}).
\end{example}

\medskip

Beyond the examples above, players might even prefer neither $\simcost = 0$ nor $\simcost = \infty$ but rather something inbetween:

\begin{example}[The preferred $\simcost$ is non-extreme]\label{ex:non-extreme}
    Informally, the underlying idea behind the example is that
        the game should have the potential for a positive-sum interaction,
        but also be unfair towards \Plone{} if they never simulate
        and unfair towards \Pltwo{} if \Plone{} always simulates.
    If we then give each player the option to opt out,
        the only way either of the players can profit is if simulation is neither free nor prohibitively expensive.
    For a detailed proof, see \Cref{ex:nontrivial-simcost} in \Cref{sec:app:proofs}.
\end{example}

%% file: literature.tex
\section{Related Work}\label{sec:related-work}

In terms of the formal framework, our work is closest to the literature on games with commitment \cite{conitzer2006computing,von2010leadership}.
    This is typically modelled as a Stackelberg game \cite{von1934marktform},
        where one player commits to a strategy
        while the other player only selects their strategy after seeing the commitment.
    In particular, \cite{letchford2014value} investigates how much the committing player can gain from committing.
    Commitment in a Stackelberg game is always observed.
        (An exception is \cite{korzhyk2011solving}, which assumes a fixed chance of commitment observation.)
        In contrast, the simulation considered in this paper would correspond to a setting where
            (1) one player pays for having \textit{the other player} make a (pure) commitment and
            (2) the latter player does not know whether their commitment is observed, as the probability of it being observed is a parameter controlled by the observer.
    Ultimately, these differences imply that
        the Stackelberg game results are highly relevant as inspiration, but they are unlikely to have immediate technical implications for our setting
            (except for when $\simcost < 0$).

In terms of motivation, the setting that is the closest to our paper is open-source game theory and program equilibria~\cite{McAfee1984,Howard1988,Rubinstein1998}, \cite[Sec.\,10.4]{Tennenholtz2004}.
In program games, two (or more) players each choose a program that will play on their behalf in the game, and these programs can read each other.
    To highlight the connection to the present paper, note that one approach to attaining cooperative play in this formalism is to have the programs simulate each other~\cite{RobustProgramEquilibrium}.
The setting of the program equilibrium literature differs from ours in two important ways.
    First, the program equilibrium literature assumes that both players have access to the other player's strategy.
        (Much of the literature addresses the difficulties of mutual simulation or analysis, e.g., see \cite{barasz2014robust,Critch2019,critch2022cooperative,oesterheld2022note} in addition to the above.)
    Second, with the exception of time discounting \cite{fortnow2009program}, the program equilibrium formalism assumes that access to the other player's code is without cost.

Another approach to simulation is game theory with translucent players \cite{halpern2018translucent}.
    This framework assumes that the players tentatively settle on some strategy from which they can deviate, but doing so has some chance of being visible to the other player.
    In our terminology, this corresponds to a setting where each player always performs free but unreliable simulation of the other player.

The literature on Newcomb's problem \cite{Nozick1969} in decision theory also studies problems in which one agent can predict another's choices. For example, Newcomb's problem itself and \textit{Parfit's hitchhiker} Parfit's hitchhiker \cite{Parfit1984,Barnes1997}]can be viewed as a variant of the trust game with simulation and various other scenarios in that literature can be viewed as zero-sum games with the ability to predict similar to Example \ref{ex:cheap-exp} \cite[Section 11]{Gibbard1981}\cite{Spencer2017,ExtractingMoneyFromCDT}]. However, in this literature the predictor is generally not strategic. Instead, the predictor's behavior is assumed to be fixed (e.g, always simulate and best respond to the observed strategy). The game-theoretic, mutually strategic nature of our simulation games is essential to the present paper. The literature on Newcomb's problem instead focuses on more philosophical issues of how an agent should choose when being the subject of prediction.


%% file: conclusions.tex
\paragraph{Summary}
In this paper,
    we considered how the traditional game-theoretic setting changes
        when one player obtains the ability to run an accurate but costly simulation of the other.
    We established some basic properties of the resulting simulation games.
        We saw that
            (between breakpoint values of which there can be only finitely many),
            their equilibria change piecewise constantly/linearly (for \Plone{}/\Pltwo{}) with the simulation cost.
        Additionally, the value of information of simulating is often equal to the simulation cost.
        These properties had strong implications for the equilibria of games with cheap simulation and allowed us to prove several deeper results.
    Our initial hope was that simulation could counter a lack of trust --- and this turned out to be true.
    However, we also saw that the effects of simulation can be ambiguous, or even harmful to both players.
        This suggests that before introducing simulation to a new setting (or changing its cost), one should determine whether doing so is likely to be beneficial or not.
    Fortunately, our analysis revealed that for the very general class of normal-form games with generic payoffs, this can be done cheaply.

\paragraph{Future Work}
The future work directions we find particularly promising are the following:
    First, the results on generic-payoff NFGs cover the normal-form representations of some, but not all, extensive-form games.
        Extending these results to EFGs thus constitutes a natural next step.
    Second, we saw that the cost of simulation that results in the socially-optimal outcome varies between games.
        It might therefore be beneficial to learn how to tailor the simulation cost to the specific game, and to what value.
    Third, we assumed that simulation predicts not only the simulated agent's policy, but also the result of any of their randomization --- i.e., their precise action.
        Whether this assumption makes sense depends on the precise setting,
            but in any case,
            by considering mixtures over \textit{behavioral} strategies \cite{halpern2021sequential},
        it might be possible to go beyond this assumption while recovering most of our results.
    Finally, our work assumes that simulation is perfectly reliable, captures all parts of the other agent, and is only available to one agent but not the other.
        Ultimately, it will be necessary to go beyond these assumptions.
        We hope that progress in this direction can be made by developing a framework that encompasses both our work and some of the formalisms discussed in \Cref{sec:related-work} (and in particular the work on program equilibria).

\edit{
\paragraph{Limitations}
The simulation approach to cooperation has various limitations.
Apart from the obstacles implied by the future work above,
there is the issue of making sure that the agent we are simulating is the same as the agent we end up interacting with
    --- for example, the other party might try to feed us fake source code, or change it after sharing it.
Moreover, the simulated party needs to be willing to its policy, which might be in tension with retaining privacy, trade secrets, etc.
Finally, the simulation approach relies on the simulated agents being unable to differentiate between simulation and reality.
    This might be difficult to achieve as the relevant situations become more complicated and AI agents grow in capability.
}

%% file: acknowledgements.tex
\section*{Acknowledgments}

We are grateful to
    Emanuel Tewolde for pointing out the connection between \Cref{prop:piecewise_constant_linear} and parametrized linear programming,
    Zuzana Kovarikova for help with \Cref{lem:limit-eq-alternative-def},
    Lewis Hammond for discussions and feedback on an earlier version of the text,
    \edit{and Emin Berker for feedback on the camera-ready version.
We would also like to thank an anonymous IJCAI reviewer for their suggestions regarding the definition of simulation games and inspiring questions regarding the difficulty of determining the usefulness of simulation.}
We thank the Cooperative AI Foundation, Polaris Ventures (formerly the Center for Emerging Risk Research) and Jaan Tallinn's donor-advised fund at Founders Pledge for financial support.

%% file: proofs.tex
\section{Proofs}\label{sec:app:proofs}

\edit{
\simgameDefEquivalence*

\begin{proof}
This immediately follows from \Cref{def:no-br-tiebreaking} and \Cref{def:simgame}.
\end{proof}

\reductionToFixedBRpolicy*

\begin{proof}
Let $\game = (\actions, \utility)$
    and suppose that $\policy \in \NE(\simgame^{\simcost,\all})$ is not in $\NE(\game)$.
This means that $\policy_1$ must put a non-zero probability on at least one of the simulate actions and we have
    $
        \lambda
        :=
        \sum_{\brSelectorPure \in \brSelectorsPure}
            \policy_1(\simulate_{\brSelectorPure})
        \in (0, 1]
        .
    $
As a result, we can define a (possibly stochastic) best-response policy
    $
        \brSelector
        :=
        \lambda^{-1}
        \sum_{\brSelectorPure \in \brSelectorsPure}
                \policy_1(\simulate_{\brSelectorPure})
            \brSelectorPure
        \in \brSelectors
    $
and use it to construct a policy $\rho$ for $\simgame^{\simcost, \brSelector}$ as
    $\rho_2 := \policy_2$,
    $\rho_1(\action) := \policy_1(\action)$ for $\action \in \actions_1$,
    $\rho_1(\simulate_{\brSelector}) := \lambda$.
By definition of $\simgame^{\simcost, \all}$ and $\simgame^{\simcost, \brSelector}$,
    the utility of any potential deviation from $\rho$ in $\simgame^{\simcost, \brSelector}$ is the same as the utility of the corresponding deviation from $\policy$ in $\simgame^{\simcost, \all}$.
This implies that $\rho$ is a $\NE$ in $\simgame^{\simcost, \brSelector}$.

In the opposite direction, any $\rho \in \NE(\simgame^{\simcost, \brSelector})$,
    corresponds to a unique $\policy$ in $\simgame^{\simcost, \all}$
--- this is because every stochastic best-response policy can be uniquely written as a convex combination of deterministic best-response policies.
Since $\rho$ is an equilibrium, \Pltwo{} will have no incentive to deviate from $\policy$ either
    (because the utilities in $\simgame^{\simcost, \all}$ would be identical to the utilities in $\simgame^{\simcost, \brSelector}$).
The same will be true for \Plone{}'s potential deviations to $\action \in \actions_1$.
Finally, while replacing $\simulate_{\brSelector}$ by some deterministic best-response policy could change the utilities of \Pltwo{},
    it will not make a difference to \Plone{}'s utilities over $\simulate_{\brSelector}$,
    so such actions cannot lead to a profitable deviation either.
This shows that $\policy$ will be a $\NE$ of $\simgame^{\simcost, \all}$.
\end{proof}
}

\basicNEresults*

\begin{proof}
In (i), the dominance claim hold because $\utility_1(\simulate, \actionOpp) = \utility_1(\br, \actionOpp) - \simcost$ for every $\actionOpp \in \actions_2$.
    As a result, when \Plone{} simulates with probability $1$,
        \Pltwo{} gets utility $\utility_2(\simulate, \actionOpp) = \utility_2(\brSelector(\actionOpp), \actionOpp)$ for some best-response policy $\brSelector$.
    As a result, \Pltwo{} must select an action which maximises this value.
    And since \Plone{} realises the utility $\utility_2(\simulate, \actionOpp)$ by playing according to $\brSelector(\actionOpp)$,
    we can identify this equilibrium with some pure Stackleberg equilibrium of $\game$ where \Pltwo{} is the leader.

In (ii), $\simulate$ is strongly dominated by \Plone{}'s min-max policy.
    In particular, $\simulate$ cannot be played in any NE.
\end{proof}

\VoIElementary*

\begin{proof}
Let $\policy$ be a policy in $\game$.
Recall that $\VoI(\policy_2)$ is defined as
    the difference between
        $\sum_{\actionOpp \in \actions_2} \policy_2(\actionOpp) \utility_1(\br, \actionOpp)$
        and \Plone{}'s best-response utility against $\policy_2$.
In other words, we have
\begin{align*}
    & \utility_1(\simulate, \policy_2)
    \\
    \ & =  \sum_{\actionOpp \in \actions_2}
                \policy_2(\actionOpp)
                \utility_1(\br, \actionOpp)
        - \simcost
    \\
    \ & = \left(
            \sum_{\actionOpp \in \actions_2}
                \policy_2(\actionOpp)
                \utility_1(\br, \actionOpp)
            - \utility_1(\br, \policy_2)
        \right)
        + \utility_1(\br, \policy_2)
        - \simcost
    \\
    \ & = \VoI(\policy_2) + \utility_1(\br, \policy_2) - \simcost 
    \\
    \ & = \utility_1(\br, \policy_2) + (\VoI(\policy_2) - \simcost )
    .
\end{align*}
\end{proof}

This result immediately yields the following:

\begin{restatable}{corollary}{tremblingHandPerfect}\label{cor:trembling-hand-perfect}
Let $\game$ be a game in which $\bigcap_{\actionOpp \in \actions_2} \br(\actionOpp) = \emptyset$.
Then all trembling-hand-perfect NE of $\simgame^0$ satisfy $\policy_1(\simulate) = 1$.
In particular, the set of trembling-hand-perfect NE of $\simgame^0$ can be identified with the set of pure Stackelberg equilibria of $\game$.
\end{restatable}

\VOIequalsSIMCOST*

\begin{proof}
(1) This proposition straightforwardly follows from Lemma~\ref{lem:VoIandSim}.
Indeed, if $\VoI(\policy_2) < \simcost$, the equation implies that deviating to $\simulate$ would decrease \Plone{}'s utility, and thus $\simulate$ cannot be in the support of $\policy_1$.
If $\VoI(\policy_2) > \simcost$, simulation would give strictly higher utility (against $\policy_2$) than any action from the original game, so simulation would have to be the \textit{only} action in the support of $\policy_1$.
Consequently, the only case when the support of $\policy_1$ can include both $\simulate$ and some other action is when $\VoI(\policy_2) = \simcost$.

(2) Suppose that
    $\policy$ is a NE of $\simgame^0$
    with $\VoI(\policy_2) > 0$.
By Lemma~\ref{lem:VoIandSim}, this means that simulating is a strongly dominant action for \Plone{} and $\policy_1(\simulate) = 1$.
Subsequently, any action from the support of $\policy_2$ must be an optimal commitment against $\brSelector$.
However, the definition of $\VoI$ implies that there can be no single action of \Plone{} which would give maximum utility against all actions $\actionOpp$ from the support of $\policy_2$.
    In other words, $\game$ must have optimal commitments of \Pltwo{} that do not share a best response.
This concludes the proof.
\end{proof}

\VoIandSpecNE*

\begin{proof}
Let $\policy$ be a NE of $\simgame$.
Proposition~\ref{lem:VoI} implies that when $\simcost < \VoI(\policy_2)$, $\policy$ cannot be a NE of $\simgame^\simcost$ (since $\simulate$ is not in the support of $\policy_1$).
Conversely, when $\simcost \geq \VoI(\policy_2)$,
    Lemma~\ref{lem:VoIandSim} implies that \Plone{} isn't incentivised to unilaterally switch to $\simulate$.
    Moreover, since $\policy$ is a NE of $\game$, no player is incentivised to switch to any other actions.
    As a result, $\policy$ is a NE of $\simgame^\simcost$ for any $\simcost \geq 0$.
\end{proof}

\limitEquilibriaWLOG*
\begin{proof}
Let $\exceptionPoint_1$ be the first breakpoint of $\simgame$ that is higher than $0$.
Let  $\policy^0$ be a limit equilibrium of $\simgame$
    and let $\policy^n_2$ be a sequence of strategies for which $\policy^0_2 = \lim_n \policy^n_2$, $(\policy^0_1, \policy^n_2) \in \NE(\simgame^{\simcost_n})$, $\simcost_n \to 0_+$.
By Proposition~\ref{prop:piecewise_constant_linear},
    each $\policy^n_2$ lies on some line segment $\NEtrajectory_n : \simcost \in [0, \exceptionPoint_1] \mapsto \policy^n_2 + \direction^n (\simcost - \simcost_n)$,
    where $\direction^n \in \R^{\actions_2}$ is the direction the line goes in
    and $(\policy^0_1, \NEtrajectory_n (\simcost)) \in \NE(\simgame^\simcost)$ for each $\simcost \in [0, \exceptionPoint_1]$.
The set $\{ \direction_n \mid n \in \N \}$ is necessarily bounded in $\R^{\actions_2}$
    (otherwise $\NEtrajectory_n(\exceptionPoint)$ would be unbounded in $\R^{\actions_2}$ --- i.e., it wouldn't lie $\policies_2$).
Using a compactness argument, we can assume that $\direction_n$ converges to some $\direction_0 \in \R^{\actions_2}$.
Denote by $\NEtrajectory_0$
    the line segment $\NEtrajectory_0 : \simcost \in [0, \exceptionPoint_1] \mapsto \policy^0_2 + \direction_0 (\simcost - 0)$.
Since the set
    $\{
        (\simcost, \policy)
        \mid
            \simcost \in [0, \exceptionPoint_1],
            \,
            \policy \in \NE(\simgame^\simcost)
    \}$
is closed,
$\NEtrajectory_0$ satisfies
    $(\policy^0_1, \NEtrajectory_0(\simcost)) \in \NE(\simgame^\simcost)$
    for every $\simcost \in [0, \exceptionPoint_1]$.
Denoting $\exceptionPoint := \exceptionPoint_1$ and $\policy^\exceptionPoint_2 := \NEtrajectory_0(\exceptionPoint)$ concludes the proof.
\end{proof}

\limitEquilibriaVoI*

\newcommand{\likesBaseline}{\actions_2^\textnormal{LB}}
\newcommand{\likesSimulation}{\actions_2^\textnormal{LS}}
\newcommand{\simProbUpper}{\bar \simProb_\deviation}
\newcommand{\simProbLower}{\underline \simProb^\deviation}
\newcommand{\simProbMax}{\bar \simProb_*}
\newcommand{\simProbMin}{\underline \simProb^*}

\begin{proof}
Let $\game$, $\simcost_0$, and $\policy$ be as in the assumptions of the lemma.
To prove the statement, we first identify $\slope$ and the policies $\policy^\baseline$ and $\policy^\deviate_2$, and then we show that they have the desired properties.

\smallskip

Finding $\policy^\baseline_1$ is trivial
    --- we simply $\policy_1 =: (1 - \policy_1(\simulate) ) \cdot \policy^\baseline_1 + \policy_1(\simulate) \cdot \simulate$ and observe that $\policy^\baseline_1$ must be a valid policy for \Plone{}.
To find $\slope$, $\policy^\baseline_2$, and $\policy^\deviate_2$, we can use \Cref{prop:piecewise_constant_linear}.
    Indeed, this proposition implies that there is some linear function $\simcost \in [0, \exceptionPoint_1] \mapsto \policy^\simcost_2 \in \policies_2$
    for which $(\policy_1, \policy^\simcost_2) \in \NE(\simgame^\simcost)$.
        (Once we define the desired properties, this proves the condition (i).)
    To define the baseline policy of \Pltwo{}, we simply set $\policy^\baseline_2 := \policy^0_2$.
    To define the deviation policy,
        we first project $\policy^{\exceptionPoint_1}_2$ onto $\policy^\baseline_2$ by setting
        $\beta := \max \{ \beta' \mid \policy^{\exceptionPoint_1}_2 - \beta' \policy^\beta_2 \geq 0 \}$ (where the inequality $\geq$ holds pointwise).
    We then set
        $\tilde \policy^\deviate_2 := \policy^{\exceptionPoint_1}_2 - \beta \policy^\baseline_2$
        and define $\policy^\deviate_2 := \tilde \policy^\deviate_2 / \lVert \tilde \policy^\deviate_2 \rVert$.
    Finally, by setting $\slope := \beta / \exceptionPoint_1$,
        we get the desired ``slope'' for which
        $(1 - \slope \simcost) \cdot \policy^\baseline_2 + \slope \simcost \cdot \policy^\deviate_2 = \policy^\simcost_2$.
        (This holds because the functions on both sides of this equation are linear and they coincide at $\simcost = 0$ and $\simcost = \exceptionPoint_1$.)

\smallskip

As a side-product of the previous paragraph, we already have (i).
To prove the lemma, it remains to prove that $\policy^\baseline$ satisfies (B1-2) and $\policy^\deviate_2$ satisfies (D1-3).

\medskip

\textbf{(ii)}
\textbf{(B1):}
    By Lemma~\ref{lem:VoI}, $\VoI(\policy^\baseline_2)$ must be $0$.
    This implies that any action in the support of $\policy_1$ must be a best-response to any action from the support of $\policy^\baseline_2$ --- i.e., we have (B1).

\textbf{(B2):}
    Since $\game$ admits no best-response tie-breaking by \Plone{},
    (B1) implies that playing any $\actionOpp$ with $\supp(\policy^\baseline_1) \subseteq \br(\actionOpp)$ is guaranteed to yield the same utility
    \begin{align*}
        \utility_2(\policy_1, \actionOpp)
        & = (1 - \policy_1(\simulate) ) \utility_2(\policy^\baseline_1, \actionOpp) + \policy_1(\simulate) \utility_2(\simulate, \actionOpp) \\
        & = (1 - \policy_1(\simulate) ) \utility_2(\br, \actionOpp) + \policy_1(\simulate) \utility_2(\br, \actionOpp) \\
        & = \utility_2(\br, \actionOpp)
        .
    \end{align*}
As a result, any $\actionOpp$ from the support of $\policy^\baseline$ must satisfy
    \begin{align*}
        \utility_2(\br, \actionOpp) = \max \{ \utility_2(\br, \actionOppAlt) \mid \br(\actionOppAlt) \supseteq \supp(\policy^\baseline_1) \}.
    \end{align*}
    Indeed, if it did not, \Pltwo{} could increase their utility by switching to an action that does satisfy this equality,
        thus contradicting the fact that $\policy^\baseline$ is a NE of $\simgame^0$.
    This shows that $\policy^\baseline$ satisfies (B2).

\medskip

\textbf{(iii):} To prove this part, consider $\policy^\simcost$ for some $\simcost \in (0, \exceptionPoint_1)$.

\textbf{(D1):} Suppose there was a single action of \Plone{} that was a best response to both all actions from $\supp(\policy^\baseline_2)$ and all actions from $\supp(\policy^\deviate_2)$.
    Then \Plone{} could gain utility by unilaterally switching to that action (since $\policy_1(\simulate) > 0$ and $\simcost > 0$).
    This in particular implies that no action from $\supp(\policy^\baseline_1)$ can have this property.

\smallskip

\textbf{(D2):}
    First, suppose some $\deviation \in \action_2$ satisfied
        both $\utility_2(\policy^\baseline_1, \deviation) > \utility_2(\policy^\baseline)$
        and $\utility_2(\br, \deviation) > \utility_2(\policy^\baseline)$.
        Then \Pltwo{} could gain utility by unilaterally deviating to $\deviation$, contradicting the fact that $\policy^\simcost$ is a Nash equilibrium.
    The same would be true if some $\deviation$ satisfied these formulas with $=$ and $>$ or with $>$ and $=$.
    Conversely, any action of \Pltwo{} that
        satisfies these formulas with $<$ and $<$, $<$ and $=$, or $=$ and $<$
        is dominated (against $\policy_1$) by playing $\policy^\baseline_2$
        and therefore cannot be played in equilibrium.

\smallskip

\textbf{(D3):}
    Denote by $\likesSimulation$, resp. $\likesBaseline$ the sets of actions that satisfy $\likesSimLabel$, resp. $\likesBaselineLabel$;
        these are the actions for which \Pltwo{} would \textbf{L}ike \Plone{} to \textbf{S}imulate, resp. would \textbf{L}ike them to play their \textbf{B}aseline strategy.
    
    First, we will consider some $\deviation \in \likesBaseline$
        and determine the simulation probability $\simProb$ that would make \Plone{} indifferent between $\policy^\baseline$ and $\deviation$.
        To determine $\simProb$, we first observe that
            $
                \utility_2(\policy_1, \policy^\baseline_2)
                    = \utility_2(\policy^\baseline_1, \policy^\baseline_2)
                    = \utility_2(\policy^\baseline)
            $
            and
            $
                \utility_2(\policy_1, \deviation)
                    = (1 - \simProb) \utility_2(\policy^\baseline_1, \deviation) + \simProb \utility_2(\br, \deviation)
            $.
        This yields
            \begin{align*}
                & \utility_2(\policy_1, \policy^\baseline_2) = \utility_2(\policy_1, \deviation) \\
                & \iff
                    \utility_2(\policy^\baseline)
                    = (1 - \simProb) \utility_2(\policy^\baseline_1, \deviation) + \simProb \utility_2(\br, \deviation) \\
                & \iff
                    \utility_2(\policy^\baseline)
                    = \utility_2(\policy^\baseline_1, \deviation)
                        - \simProb \left(
                            \utility_2(\policy^\baseline_1, \deviation) - \utility_2(\br, \deviation)
                        \right) \\
                & \iff
                    \simProb
                    = \frac{
                        \utility_2(\policy^\baseline_1, \deviation) - \utility_2(\policy^\baseline)
                    }{
                        \utility_2(\policy^\baseline_1, \deviation) - \utility_2(\br, \deviation)
                    } \\
                & \iff
                    \simProb
                    = \frac{
                        \utility_2(\policy^\baseline_1, \deviation) - \utility_2(\policy^\baseline)
                    }{
                        \left( \utility_2(\policy^\baseline_1, \deviation) - \utility_2(\policy^\baseline) \right)
                        + \left( \utility_2(\policy^\baseline) - \utility_2(\br, \deviation) \right)
                    }
                .
            \end{align*}
        Denote the right-hand side of the last line as
        \begin{align}\label{eq:sim-prob-upper-bound}
                \simProbUpper
                := 
                \frac{
                    \utility_2(\policy^\baseline_1, \deviation) - \utility_2(\policy^\baseline)
                }{
                    \left( \utility_2(\policy^\baseline_1, \deviation) - \utility_2(\policy^\baseline) \right)
                    + \left( \utility_2(\policy^\baseline) - \utility_2(\br, \deviation) \right)
                }
            .
        \end{align}
        Clearly, $\simProbUpper$ is a strictly increasing function of the deviation attractiveness ratio
        \begin{align*}
            \deviationRatio_\deviation
            = \frac{
                \utility_2(\policy^\baseline_1, \deviation) - \utility_2(\policy^\baseline)
            }{
                \utility_2(\policy^\baseline) - \utility_2(\br, \deviation
            }
            .
        \end{align*}
            (Intuitively, $\deviationRatio_\deviation$ captures the tradeoffs \Pltwo{} faces when deviating and hoping they will not be caught by the simulator.)
        Since $\deviation$ is of the type that causes \Pltwo{} to prefer $\policy^\baseline_1$ over simulation, this implies that \Pltwo{} would
            deviate to $\deviation$ for $\policy_1(\simulate) < \simProbUpper$,
            be indifferent for $\policy_1(\simulate) = \simProbUpper$,
            and switch to $\deviate$ for $\policy_1(\simulate) > \simProbUpper$.
        Finally, denote $\simProbMax := \max \{ \simProbUpper \mid \deviation \in \likesBaseline \}$.
    
    Considering the same equation for $\deviation \in \likesSimulation$, we get
        \begin{align}
              & \utility_2(\policy_1, \policy^\baseline_2) = \utility_2(\policy_1, \deviation)
                \iff \dots
                \iff \simProb = \simProbLower, \textnormal{ where}
                \nonumber \\
                & \simProbLower
                    := \frac{
                        \utility_2(\policy^\baseline) - \utility_2(\policy^\baseline_1, \deviation)
                    }{
                        \left( \utility_2(\policy^\baseline) - \utility_2(\policy^\baseline_1, \deviation) \right)
                        + \left( \utility_2(\br, \deviation) - \utility_2(\policy^\baseline) \right)
                    }
                    \label{eq:sim-prob-lower-bound}
                .
        \end{align}
        Clearly $\simProbLower$
             is a strictly decreasing function of inverse ratio
             \begin{align*}
                \deviationRatioInverse_\deviation
                = \frac{
                    \utility_2(\br, \deviation - \utility_2(\policy^\baseline))
                }{
                    \utility_2(\policy^\baseline) - \utility_2(\policy^\baseline_1, \deviation)
                }
                .
             \end{align*}
            (In contrast to $\deviationRatio_\deviation$, this ratio captures the tradeoffs \Pltwo{} faces when deviating and hoping they \textit{will} be caught by the simulator.)
        Finally, denote $\simProbMin := \min \{ \simProbLower \mid \deviation \in \likesSimulation \}$.

\smallskip

    Using these calculations, we are not only able to conclude the proof, but we have in fact also determined the values of $\policy_1(\simulate)$ that are compatible with $\policy^\baseline$:
        If $\policy_1(\simulate)$ was strictly lower than $\simProbMax$,
            \Pltwo{} would deviate towards some $\deviation \in \likesBaseline$ for which $\simProbUpper > \policy_1(\simulate)$.
        If it was strictly higher than $\simProbMin$,
            \Pltwo{} would deviate towards some $\deviation \in \likesSimulation$ for which $\simProbUpper < \policy_1(\simulate)$.
        (In particular, we must have $\simProbMax \leq \simProbMin$ --- otherwise, $\policy^\baseline$ could not be a limit equilibrium in $\simgame$.)
        This shows that if there is some action that satisfies $\indiffLabel$, $\policy_1(\simulate)$ can take any value from $[\simProbMax, \simProbMin]$.
        For $\supp(\policy^\deviate_2)$ to contain some action $\deviation \in \likesBaseline$,
            $\policy_1(\simulate)$ must be equal to $\simProbMax$ and $\deviation$ must satisfy $\simProbUpper = \simProbMax$.
            (Which gives the ``$>$'' part of (D3).)
        And analogously, for $\supp(\policy^\deviate_2)$ to contain some action $\deviation \in \likesSimulation$,
            $\policy_1(\simulate)$ must be equal to $\simProbMin$ and $\deviation$ must satisfy $\simProbLower = \simProbMin$.
            (Which gives the ``$<$'' part of (D3).)    
    This concludes the whole proof.
\end{proof}

\supportSizeTwo*

\begin{proof}
First, we observe several implications of the assumption that $\game$ has generic payoffs.
\begin{itemize}
    \item[(0)] In a generic game, no two payoffs are the same.
    \item[(1)] For every action $\actionOpp$ of \Plone{},
        \Plone{} only has a single best response.
        With a slight abuse of notation, we denote this action as $\br(\actionOpp)$.
    \item[(1')] The same applies for \Pltwo{}.
    \item[(2)] By (1), there is a unique action of \Pltwo{} for which
        \begin{align}
            \utility_2(\br(\actionOpp), \actionOpp)
            = \max_{\actionOppAlt \in \actions_2}
                \utility_2(\br(\actionOppAlt), \actionOppAlt)
            .
            \label{eq:opt-commitment}
        \end{align}
    \item[(3)] Any two distinct actions $\deviation_1$, $\deviation_2$ of \Pltwo{}
        must also have distinct attractiveness ratios $\deviationRatio_\deviation$, $\deviationRatio_{\deviation'}$ from (D3) of \Cref{lem:limit-NE-structure}.
        (Indeed, if the payoffs of $\game$ are i.i.d. samples from the uniform distribution over $[0, 1]$, the probability two of these ratios coinciding is $0$.)
    \item[(3')] From (3), it further follows that the variables $\simProbUpper$ and $\simProbLower$,
        defined in equations \eqref{eq:sim-prob-upper-bound} and \eqref{eq:sim-prob-lower-bound},
        will also differ for different actions.
\end{itemize}

We now proceed with the proof by separately considering the cases $\policy_1(\simulate) = 1$, $\policy_1(\simulate) = 0$, and $\policy_1(\simulate \in (0, 1)$.

\smallskip

$\policy_1(\simulate) = 1$:
    If \Plone{} simulated with probability $1$, \Pltwo{} could respond be playing actions that satisfy \eqref{eq:opt-commitment}.
    By (2), there is only one such action; call it $\actionOpp$.
    However, this would mean that \Plone{} could gain additional $\simcost$ utility by switching from $\simulate$ to $\br(\actionOpp)$,
        contradiction the assumption that $\policy$ is an equilibrium.
    As a result, a generic game with cheap simulation will never have an equilibrium where \Plone{} simulates with probability $1$.

\smallskip

$\policy_1(\simulate) = 0$:
    Suppose that $\policy$ is a NE of $\simgame^{\simcost_0}$ for some $\simcost_0 \in (0, \exceptionPoint_1)$.
    We will show that $\policy$ must be pure.
    
    Let $\simcost \in [0, \exceptionPoint_1] \mapsto \policy^\simcost_2$ be some linear function (given by \Cref{prop:piecewise_constant_linear})
        for which $(\policy_1, \policy^\simcost_2) \in \NE(\simgame^\simcost)$ holds for every $\simcost$.
    Since the $\simulate$ is not in the support of $\policy_1$, $\VoI(\policy^0_2)$ must be equal to $0$
        (otherwise \Plone{} could gain by deviating to $\simulate$ for $\simcost = 0$, and $(\policy_1, \policy^0_2)$ would not be a NE of $\simgame^0$).
    This means that there must exist some $\action \in \actions_1$ that is a best-response to every $\actionOpp$ from $\supp (\policy^0_2)$.
        However, recall that (1) implies that \Plone{} only has a single best response for every action of \Pltwo{}.
        As a result, $\action$ is the \textit{only} action in the support of $\policy_1$.
    By (1'), this means that for every $\simcost \in [0, \exceptionPoint_1]$ --- and for $\simcost_0$ in particular --- $\policy^\simcost$ must also be pure.

\smallskip

$\policy_1(\simulate) \in (0, 1)$:
    Let $\policy$ be a NE of $\simgame^{\simcost_0}$ for some $\simcost_0 \in (0, \exceptionPoint_1)$.
    By \Cref{lem:limit-NE-structure}, $\policy$ can be expressed as a convex combination of
        some baseline policy $\policy^\baseline_1$ and simulation (for \Plone{}),
        resp. of $\policy^\baseline_2$ and some deviation policy $\policy^\deviate_2$ (for \Pltwo{}).
    Combining the condition (B1-2) from \Cref{lem:limit-NE-structure} with (1) and (2), we get that $\policy^\baseline$ must be pure.

    Let $\deviation$ be some element of $\supp(\policy^\deviate_2)$ and consider the three cases listed in (D2).
        If $\deviation$ satisfies $\indiffLabel$, (0) implies that it must be equal to $\policy^\baseline_2$,
            and thus not count against the size of $\supp(\policy_2)$.
            Moreover, to avoid contradicting (D1), $\supp(\policy_2)$ must also contain some other action that does not satisfy $\indiffLabel$.
        If $\deviation$ satisfies $\likesBaselineLabel$ or $\likesSimLabel$,
            the probability $\policy_1(\simulate)$ must be equal to $\simProbUpper$, resp. $\simProbLower$.
            (We observed this in the last paragraph of the proof of \Cref{lem:limit-NE-structure}.)
        By (3'), it is impossible for this to be true for two different actions $\deviation' \neq \deviation$ at the same time.
    Together with $\policy^\baseline$ being pure, this shows that $| \supp(\policy_2) | = 2$ and concludes the proof.
\end{proof}

\journal{
\newcommand{\harmonicMean}{H}
\newcommand{\indexSubset}{I}
\newcommand{\payoffA}{x}
\newcommand{\payoffB}{y}
\newcommand{\payoffsA}{\vec \payoffA}
\newcommand{\payoffsB}{\vec \payoffB}
\newcommand{\payoffSubsetA}{\payoffsA_\indexSubset}
\newcommand{\payoffSubsetB}{\payoffsB_\indexSubset}
\newcommand{\prodA}{\mathbf{\payoffA}}
\newcommand{\prodB}{\mathbf{\payoffB}}

\begin{figure*}[!tb]
    \centering
    \includegraphics[width=0.24\textwidth]{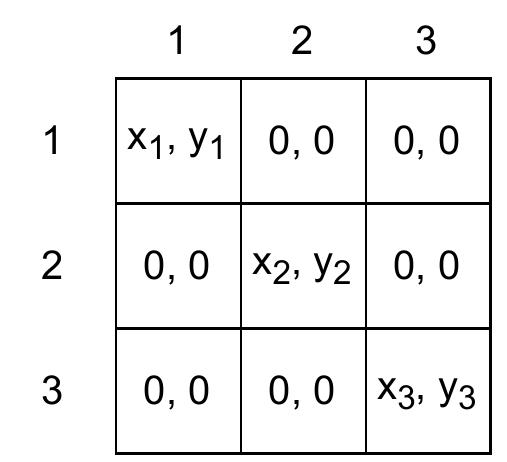}
    \caption{
        Cafés in Paris: Alice and Bob want to meet but they need to coordinate on which café to go to.
        We assume that $\payoffA_i, \payoffB_i > 0$ for every $i$.
        The actual game has $n \in \N$ actions.
    }
\end{figure*}

\exponentialBreakpoints*

\begin{proof}
Let
    $n \in \N$ and
    let $\game$ be the Cafés in Paris coordination game with payoffs $\payoffsA = (\payoffA_1, \dots, \payoffA_n)$, $\payoffsB = (\payoffB_1, \dots, \payoffB_n)$, $\payoffA, \payoffB > 0$
        (\Cref{fig:cafes-in-paris}).
We will show that for every $\indexSubset \subseteq \{1, \dots, n\}$,
    $\game$ has a Nash equilibrium with $\VoI(\policy_2) = (1 - \nicefrac{1}{|\indexSubset|}) \harmonicMean(\payoffSubsetA)$,
    where $\harmonicMean( \symbolPlaceholder )$ denotes the harmonic mean
        \begin{align*}
            \harmonicMean(a_1, \dots, a_k)
            :=
            \left(
            \frac{\frac{1}{a_1} + \dots + \frac{1}{a_k}}{k}
            \right)^{-1}
        \end{align*}
        and $\payoffSubsetA := \left( \payoffA_i \mid i \in \indexSubset \right)$.
Once we have proved this result, the conclusion immediately follows from \Cref{prop:VoI-and-specific-NE},
    since for a generic choice of $\payoffsA$, different subsets $\indexSubset$ yield different values $(1 - \nicefrac{1}{|\indexSubset|}) \harmonicMean(\payoffSubsetA)$.

\smallskip

Let $\indexSubset \subseteq \{1, \dots, n \}$ and denote
\begin{align*}
    \prodA & := \ \ \prod_{j \in \indexSubset} \payoffA_j
    & \prodB & := \ \ \prod_{j \in \indexSubset} \payoffB_j
    .
\end{align*}
First, note that the following formula defines a Nash equilibrium with $\supp(\policy_1) = \supp(\policy_2) = \indexSubset$:
\begin{align*}
    \prodA_{-i} & := \prod_{j \in \indexSubset, \, j\neq i} \payoffA_j
    & \prodB_{-i} & := \prod_{j \in \indexSubset, \, j\neq i} \payoffB_j
    \\
    \policy_1(i) & \ \propto \ \ \prodB_{-i}
    & \policy_2(i) & \ \propto \ \ \prodA_{-i}
    .
\end{align*}
Indeed, this formula ensures that \Plone{}'s payoff for playing any $i \in \indexSubset$ is $\payoffA_i \prodA_{-i} = \prodA$, and similarly for \Pltwo{}.
Calculating the expected utility reveals that this strategy yields
    \begin{align*}
        \utility_1(\policy)
        & = \sum_{i \in \indexSubset}
                \sum_{j \in \indexSubset}
                    \policy_1(i) \policy_2(j) \utility_1(i,j)
        \\
        & = \sum_{i \in \indexSubset}
                    \policy_1(i) \policy_2(i) \utility_1(i,i)
        \\
        & = \sum_{i \in \indexSubset}
                \frac{\prodB_{-i}}{\sum_{i' \in \indexSubset} \prodB_{-i'}}
                \frac{\prodA_{-i}}{\sum_{i' \in \indexSubset} \prodA_{-i'}}
                \payoffA_i
        \\
        & = \sum_{i \in \indexSubset}
                \frac{\prodB_{-i}}{\sum_{i' \in \indexSubset} \prodB_{-i'}}
                \frac{\prodA}{\sum_{i' \in \indexSubset} \prodA_{-i'}}
        \\
        & = \sum_{i \in \indexSubset}
                \frac{1}{\sum_{i' \in \indexSubset} \frac{\prodB_{-i'}}{\prodB_{-i}}}
                \frac{1}{\sum_{i' \in \indexSubset} \frac{\prodA_{-i'}}{\prodA}}
        \\
        & =
            \sum_{i \in \indexSubset}
                \frac{1}{
                    \frac{|\indexSubset|}{|\indexSubset|}
                    \sum_{i' \in \indexSubset} \frac{\payoffB_i}{\payoffB_{-i'}}
                }
            \frac{1}{
                \frac{|\indexSubset|}{|\indexSubset|}
                \sum_{i' \in \indexSubset} \frac{1}{\payoffA_{i'}}
            }
        \\
        & =
            \frac{1}{|\indexSubset|}
            \frac{1}{
                \frac{1}{|\indexSubset|}
                \sum_{i' \in \indexSubset} \frac{1}{\payoffA_{i'}}
            }
            \sum_{i \in \indexSubset}
                \frac{1}{
                    \payoffB_i
                    \frac{|\indexSubset|}{|\indexSubset|}
                    \sum_{i' \in \indexSubset} \frac{1}{\payoffB_{-i'}}
                }
        \\
        & =
            \frac{1}{|\indexSubset|} \harmonicMean(\payoffSubsetA)
            \frac{1}{
                \frac{1}{|\indexSubset|}
                \sum_{i' \in \indexSubset} \frac{1}{\payoffB_{-i'}}
            }
            \frac{1}{|\indexSubset|}
            \sum_{i \in \indexSubset}
                \frac{1}{\payoffB_i}
        \\
        & =
            \frac{1}{|\indexSubset|} \harmonicMean(\payoffSubsetA)
            \harmonicMean(\payoffSubsetB)
            \harmonicMean(\payoffSubsetB)^{-1}
        \\
        & = \frac{1}{|\indexSubset|} \harmonicMean(\payoffSubsetA)
        .
    \end{align*}
However, if \Plone{} was able to best-respond after seeing \Plone{}'s choice of action, they would instead get
    \begin{align*}
        \sum_{i \in \indexSubset}
            \policy_2(i) \utility_1(\br,i)
        & = \sum_{i \in \indexSubset}
                \policy_2(i) \utility_1(i,i)
        = \sum_{i \in \indexSubset}
                \frac{\prodA_{-i}}{\sum_{i' \in \indexSubset} \prodA_{-i'}}
                \payoffA_i
        \\
        & = \sum_{i \in \indexSubset}
                \frac{\prodA}{\sum_{i' \in \indexSubset} \prodA_{-i'}}
        = | \indexSubset |
                \frac{1}{\sum_{i' \in \indexSubset} \frac{\prodA_{-i'}}{\prodA}}
        \\
        & = \frac{1}{
                \frac{1}{| \indexSubset |}
                \sum_{i' \in \indexSubset} \frac{\prodA_{-i'}}{\prodA}
            }
        = \harmonicMean(\payoffSubsetA)
        .
    \end{align*}
By definition of $\VoI$ (\Cref{def:VoI}), this shows that
    $
        \VoI(\policy)
        =
        \harmonicMean(\payoffSubsetA)
        (1 - \nicefrac{1}{|\indexSubset|})
    $.
This concludes the proof
\end{proof}
}

\GsimEquallyHard*
\begin{proof}
This trivially follows from the assumption that simulation games are modelled as the original normal-form game $\game$ with the added simulate action $\simulate$.
\end{proof}

\complexityForExtremeSimcost*
\begin{proof}
(i): First, suppose that $\game$ is a game with no best-response tie-breaking (i.e., \Plone{}'s choice of best response never affects \Pltwo{}'s utility).
By Proposition~\ref{prop:NE-for-extreme-simcost}(iii), simulation strongly dominates all other actions when $\simcost < 0$.
Consequently, 
    all that is needed to solve $\simgame^\simcost$ is
    for \Pltwo{} to search through $\action_2$ for the action $\actionOpp$ with the highest best-response value $\utility_2(\br, \actionOpp) = \utility_2(\simulate, \actionOpp)$.
    As a result, the complexity of solving $\simgame^\simcost$ is dominated by the complexity of determining the best-response utilities corresponding the simulate action (which is $\bigO(|\actions|)$).

If $\game$ allows best-response tie-breaking for \Plone{}, the complexity might be higher because \Plone{} could have multiple ways of responding after simulation.
However, for the purpose of the paper, we were assuming that this policy (for how to respond after simulation) is fixed. As a result, the argument from the previous paragraph applies to this case as well.

(ii): By Proposition~\ref{prop:NE-for-extreme-simcost}(i), $\simulate$ will never be played (in a NE) for high enough $\simcost$.
    As a result, solving $\simgame$ becomes equivalent to solving $\game$.
\end{proof}

\genericGames*

\begin{proof}
By \Cref{thm:generic-games-support-size}, all NE of $\simgame^\simcost$ are either pure or have $|\supp(\policy_1)| = | \supp(\policy_2)| = 2$.
    (This straightforwardly implies that we could find all NE of $\simgame^\simcost$ in $\bigO(|\actions|^2)$ time, by trying all possible supports of size one and two.
    The purpose of the theorem is, therefore, to show that the task can even be done in linear time.)

First, note that
    since $\game$ has generic payoffs,
    $\simgame^\simcost$ doesn't have any equilibria with $\policy_1(\simulate) = 1$
    and all of its equilibria with $\policy_1(\simulate) = 0$ are pure.
    (This is not hard to see directly. For a detailed argument, see the proof of \Cref{thm:generic-games-support-size}.)
As a result, these two cases can be handled in $\bigO(|\actions|)$ time.\footnotemark
    \footnotetext{
        Recall that to find all pure NE of an NFG in linear time, we can:
            First, find all best-responses of \Plone{} to every action of \Pltwo{}.
            Then find all best-responses of \Pltwo{} to every action of \Plone{}.
            And finally use these findings to identify all joint actions that form a mutual best response;
            these coincide with all pure NE.
    }

Second, consider the case when $\policy \in \NE(\simgame^\simcost)$ satisfies $\policy(\simulate) \in (0, 1)$.
    From \Cref{thm:generic-games-support-size}, we know that
        \Plone{} will be mixing between some $\action$ and $\simulate$ and
        \Pltwo{} will be mixing between some $\actionOpp$ and $\deviation$,
        where $(\action, \actionOpp)$ is the baseline strategy satisfying (B1-2) from \Cref{lem:limit-NE-structure}
        and $\deviation$ is the deviation strategy satisfying (D1-3) from \Cref{lem:limit-NE-structure}.
    If a triplet $(\action, \actionOpp, \deviation)$ satisfies (B1-2) and (D1-3), we will call it ``suitable''.

To find all NE of $\simgame^\simcost$, we can use the following procedure:
    (1a) For each $\action \in \actions_1$, find the (unique) best response of \Pltwo{}.
    (1b) For each $\actionOpp \in \actions_2$, find the (unique) best response of \Plone{}.
    (2) Find all suitable triplets $(\action, \actionOpp, \deviation)$.
    (3) For every suitable triplet from (2), find the unique NE with $\supp(\policy_1) = \{ \action, \simulate \}$ and $\supp (\policy_2) = \{ \actionOpp, \deviation \}$,
        or learn that no such NE exists.
        (The uniqueness follows from (D1) and (D2).)
To prove that the steps (1-3) can be performed in $\bigO(|\actions|)$ time, we will use the following claims:
    (I) Each of the steps (1a) and (1b) can be performed in $\bigO(|\actions|)$.
    (II) There are at most $2 \cdot \min \{ |\actions_1|, |\actions_2| \}$ suitable,
        and it is possible to find all of them in $\bigO(|\actions|)$ time.
    (III) Performing (3) for a single suitable triplet takes $\bigO(|\actions_1| + |\actions_2|)$ time.

Clearly, the combination of (I), (II), and (II) yields the conclusion of the theorem.
    Moreover,
        (I) is elementary and
        (III) follows from the fact that performing (2) only requires
            solving the 2-by-2 game with actions $\{\action, \simulate\} \times \{\actionOpp, \deviation \}$
            and checking that none of the remaining actions is a profitable deviation.    
To prove the theorem, it remains to prove the claim (II).
We do this in two steps:
    First, we show that there are at most $\min \{ |\actions_1|, |\actions_2| \}$ pairs $(\action, \actionOpp)$ that might be a part of some suitable triplet $(\action, \actionOpp, \deviation)$.
    Second, we show that for any pair $(\action, \actionOpp)$, there are at most two actions for which the triplet $(\action, \actionOpp, \deviation)$ is suitable.
        
For the first step, note that
    for every $\actionOpp$, the only pair $(\action, \actionOpp)$ that might satisfy the condition (B1) is $(\br(\actionOpp), \actionOpp)$.
    Therefore, there are at most $|\actions_2|$ pairs $(\action, \actionOpp)$ that might be a part of some suitable triplet $(\action, \actionOpp, \deviation)$.
    Moreover, for every $\action$, there will only be a single $\actionOpp$ that satisfies the condition (B2)
        (i.e., the condition that $\actionOpp$ maximizes $\utility_2(\action, \actionOppAlt)$ among the actions $\actionOppAlt$ for which $\action \in \br(\actionOppAlt)$).
        Therefore, there are at most $|\actions_1|$ pairs $(\action, \actionOpp)$ that might be a part of some suitable triplet $(\action, \actionOpp, \deviation)$.
    Combining the two bounds shows that the number of pairs that might be a part of a suitable triplet is $\min \{ |\actions_1|, |\actions_2| \}$.

For the second step,
    note that in a generic game, the only action that satisfies $\indiffLabel$ for $\policy^\baseline = (\action, \actionOpp)$ is $\actionOpp$.
    Moreover, the genericity of $\game$ implies that
        either the set of actions satisfying $\likesBaselineLabel$ is empty,
        or there is exactly one action $\deviation$ that satisfies $\likesBaselineLabel$ and maximizes the attractiveness $\deviationRatio_\deviation$ ratio from (D3).
    Analogously, there will be at most one action $\deviation$ that satisfies $\likesSimLabel$ and maximizes the inverse attractiveness ratio $\deviationRatioInverse_\deviation$ from (D3).
    This shows that for any $(\action, \actionOpp)$ from the previous step, there are at most two actions for which the triplet $(\action, \actionOpp, \deviation)$ is suitable.
    Since this proves completes the proof of (II), we have concluded the whole proof.
\end{proof}

\journal{
\newcommand{\deviationProb}{\alpha \simcost}
\newcommand{\gameModif}{\widetilde \game}
\newcommand{\actionsModif}{\widetilde \actions}
\newcommand{\simgameModif}{\widetilde \game_\simSubscript}

\simgameHardness*

\begin{proof}
Let $\game$ be a NFG.
We will assume that $\game$ has utilities in $[0, 1]$ --- otherwise, we can rescale the utilities without changing the sets of NE.
Let $\gameModif$ be as in \Cref{fig:trust-game-hardness} and fix some $\simcost \in (0, \nicefrac{1}{3})$.
We will show that a strategy $\policy$ in $\gameModif$ with $\policy_1(\simulate) > 0$ is a NE of $\simgameModif^\simcost$
    if and only if
        it has \Plone{} mixing between simulating and $(\trust, \policy^\game_1)$
        and \Pltwo{} mixing between cooperating and $(\defect, \policy^\game_2)$,
        where $\policy^\game \in \NE(\game)$
        and $\policy_1(\simulate) = TODO$,
        $\policy_2(\defect) = TODO$.
    ($\simgameModif^\simcost$ also has ``non-simulation'' equilibria where \Plone{} always walks out and \Pltwo{} defects with high probability. These will not be relevant for our proof.)

\smallskip

Before proceeding with the proof,
    observe that for any $\policy^\game$ in $\game$, we have
    \begin{align*}
        & \utility_2(\trust, \policy^\game_1, \defect, \policy^\game_2) > \utility_2(\trust, \cooperate) \\
        & \utility_1(\trust, \policy^\game_1, \defect, \policy^\game_2) < \utility_1(\walkOut)
        .
    \end{align*}
As a result,
    if \Pltwo{} attempts to defect while \Plone{} simulates, \Plone{} will best-respond by walking out,
    which results in the worst possible utility for \Pltwo{}.
This shows that for any $\policy^\game$ in $\game$, we have
    \begin{align*}
        & \utility_1(\simulate, \cooperate) < \utility_1(\trust, \cooperate)
            & \textnormal{\& }
            & \utility_1(\simulate, \defect) > \utility_1(\trust, \policy^\game_1, \defect, \policy^\game_2)
        \\
        & \utility_2(\simulate, \cooperate) > \utility_1(\simulate, \defect, \policy^\game_2)
            & \textnormal{\& }
            & \utility_2(\trust, \cooperate) < \utility_2(\trust, \policy^\game_1, \defect, \policy^\game_2)
        .
    \end{align*}
    (Indeed,
        the first inequality holds because $\trust$ is already a best-response to $\cooperate$ and $\simcost > 0$.
        The second holds because no matter what strategy is used in $\game$, \Plone{} would be willing to pay $1$ utility to walk out instead of playing $\game$.
        The third holds by our initial observation,
        and finally the fourth holds because no matter what strategy is used in $\game$, \Pltwo{} always prefers playing $\game$ to cooperating.)
These inequalities show that for any fixed $\policy^\game$,
    \Plone{} can make \Pltwo{} indifferent between $\cooperate$ and $(\defect, \policy^\game_2)$ by mixing between $\simulate$ and $(\trust, \policy^\game_1)$,
    and \Pltwo{} can make \Plone{} indifferent between $\simulate$ and $(\trust, \policy^\game_1)$ by mixing between $\cooperate$ and $(\defect, \policy^\game_2)$.
Moreover, the probabilities for which the players are indifferent will be unique
    --- we denote them as $\simProb := \simProb(\policy^\game)$ for \Plone{}
        and $\alpha \simcost := \slope(\policy^\game, \simcost) \cdot \simcost$ for \Pltwo{}.\footnote{
        The exact formulas can be derived as in the proof of \Cref{thm:generealised-TG-no-tiebreaking}.
        However, they only become relevant for the ``moreover'' part of the proposition, so we avoid deriving them for now.
    }

\medskip

To prove the main part of the proposition, we now show that $\NE(\simgameModif^\simcost)$ can be identified with $\NE(\game)$.

\smallskip

Let $\policy$ be a NE of $\simgameModif^\simcost$ for which $\policy_1(\simulate) > 0$.
We will show that $\policy$ is of the form
    \begin{align*}
        \policy_1 & = (1 - \simProb) \cdot (\trust, \policy^\game_1) + \simProb \cdot \simulate \\
        \policy_2 & = (1 - \alpha \simcost) \cdot \cooperate + \alpha \simcost \cdot (\defect, \policy^\game_2)
        ,
    \end{align*}
    where
        $\policy^\game$ is a NE of $\game$
        and $\simProb = \simProb(\policy^\game)$, $\alpha = \alpha(\policy^\game, \simcost)$.
(1) First, note that \Pltwo{} must play both $\cooperate$ and $\defect$.
    Indeed, if \Pltwo{} only cooperated, \Plone{} would strictly prefer $\trust$ over $\simulate$.
    Conversely, if \Pltwo{} only defected, \Plone{} would strictly prefer $\walkOut$ over $\simulate$.
(2) Second, we show that \Plone{} is indifferent between $\simulate$ and $\trust$ and strictly prefers both of these actions over $\walkOut$.
    Indeed, if $\policy_1$ was such that $\policy_1(\trust) = 0$ (and $\policy_1(\simulate) > 0)$, \Pltwo{} would strictly prefer $\cooperate$ over $\defect$, contradicting (1).
    If \Plone{} was indifferent between all three actions,
        $\policy$ would have to satisfy $\utility_1(\simulate, \policy_2) = \utility_1(\trust, \policy^\game_1, \policy_2) = \utility_1(\walkOut)$.
        However, solving these equations in particular gives $\simcost = \frac{2-\utility_1^\game(\policy^\game)}{3+\utility_1^\game(\policy^\game)}$, which is impossible for $\simcost < \nicefrac{1}{3}$.
(3) Third, by the initial observation about indifference, we know that $\policy$ must be of the desired form --- as a result, it only remains to show that $\policy^\game$ is a NE of $\game$.
(4) Finally, suppose that $\policy^\game$ was not a Nash equilibrium of $\game$.
    Since $\policy$ is of the desired form, we have
    \begin{align*}
        \utility_1(\policy)
        = & \ \policy_1(\simulate) \utility_1(\simulate, \policy_2)
            + \policy_1(\trust) \policy_2(\cooperate) \cdot 3
            \\
            & + \policy_1(\trust) \policy_2(\defect) \utility^\game_1(\policy^\game)
        .
    \end{align*}
    Since only the third term depends on $\policy^\game$,
        we see that if $\policy^\game_1$ was not a best response to $\policy^\game_2$ in $\game$,
        \Plone{} could unilaterally increase their utility by replacing $(\trust, \policy^\game_1)$ by $(\trust, \br^\game(\policy^\game_2)$.
    Analogously, we have
    \begin{align*}
        \utility_2(\policy)
        = & \ \policy_2(\simulate) \utility_2(\simulate, \policy_2)
            + \policy_1(\trust) \policy_2(\cooperate) \cdot (-1)
            \\
            & + \policy_1(\trust) \policy_2(\defect) \utility^\game_2(\policy^\game)
        ,
    \end{align*}
    where only the third term depends on $\policy^\game$,
    and therefore $\policy^\game_2$ must be a best-response to $\policy^\game_1$ in $\game$.
This shows that $\policy^\game$ must be a NE of $\game$
    and concludes the proof of ``$\NE(\simgameModif^\simcost) \subseteq \NE(\game)$''.

\smallskip

Conversely, let $\policy^\game \in \NE(\game)$.
To prove the ``$\NE(\game) \subseteq \NE(\simgameModif^\simcost)$'' part of the proposition, it suffices to show that
    \begin{align*}
        \policy_1 & := (1 - \simProb(\policy^\game)) \cdot (\trust, \policy^\game_1) + \simProb(\policy^\game) \cdot \simulate \\
        \policy_2 & := (1 - \alpha(\policy^\game, \simcost) \simcost) \cdot \cooperate + \alpha(\policy^\game, \simcost) \simcost \cdot (\defect, \policy^\game_2)
    \end{align*}
    is a NE of $\simgameModif^\simcost$.
    However, the calculations performed so far already imply that this choice of $\policy$
        makes \Plone{} indifferent between $\simulate$ and $\trust$ and \Pltwo{} between $\cooperate$ and $\defect$,
        and gives no player an incentive to replace $\policy^\game_\pl$ by any other policy in $\game$.
    To show that $\policy$ is a NE of $\simgameModif^\simcost$, it remains to show that \Plone{} has no incentive to deviate by playing $\walkOut$ --- in other words, that $\utility_1(\policy) \geq 2$.
        We already know that \Plone{} cannot be indifferent about $\walkOut$ (since this only happens for $\simcost = \frac{2-\utility_1^\game(\policy^\game)}{3+\utility_1^\game(\policy^\game)}$.
        Therefore, since $\simcost < \nicefrac{1}{3}$, it suffices to show that $\alpha(\policy^\game, \simcost)$ does not depend on $\simcost$ (and therefore $\policy_2(\defect) \to 0$ as $\simcost \to 0_+$).
        Since we will obtain this independence of $\alpha$ on $\simcost$ as a by-product of the last step of our proof,
            the proof of the main part of the proposition is finished for now.

\medskip

To conclude the whole proof, it remains to show the ``moreover'' part of the proposition.
We will do this by deriving the exact formulas for $\simProb(\policy^\game)$ and $\alpha(\policy^\game, \simcost)$.
(This part of the proof is essentially identical to the calculations from the proof of \Cref{thm:generealised-TG-no-tiebreaking}.)

\noindent
For $\policy_1(\simulate) = \simProb(\policy^\game) = \simProb$, we have
     \begin{align*}
        & \utility_2(\policy_1, \cooperate) = \utility_2(\policy_1, \defect, \policy^\game_2) \\
        & \iff
            -1
            = (1 - \simProb) \utility_2(\trust, \policy^\game_1, \defect, \policy^\game_2) + \simProb \utility_2(\simulate, \defect, \policy^\game_2) \\
        & \iff
            -1 
            = \utility^\game_2(\policy^\game)
                - \simProb \left(
                    \utility^\game_2(\policy^\game) - (-2)
                \right) \\
        & \iff
            \simProb
            = \frac{
                \utility^\game_2(\policy^\game) + 1
            }{
                \utility^\game_2(\policy^\game) + 2
            }
            \\
        & \iff
            \utility^\game_2(\policy^\game)
            = \frac{
                1
            }{
                \simProb + 1
            }
        .
    \end{align*}

\noindent
For $\policy_2(\defect) = \alpha(\policy^\game, \simcost) \simcost = \alpha \simcost$, we have
    \begin{align*}
        & \utility_1(\trust, \policy^\game_1, \policy_2) = \utility_1(\simulate, \policy_2) \\
        & \iff
            (1 - \deviationProb) \utility_1(\trust, \policy^\game_1, \cooperate) + \deviationProb \utility_1(\trust, \policy^\game_1, \defect, \policy^\game_2)
              \\ & \phantom{\iff} \ \ 
                 = (1 - \deviationProb) \utility_1(\trust, \policy^\game_1, \cooperate)
                     + \deviationProb \utility_1(\walkOut, \defect, \policy^\game_2) - \simcost \\
        & \iff
            \deviationProb \utility_1(\trust, \policy^\game_1, \defect, \policy^\game_2)
                = \deviationProb \utility_1(\walkOut, \defect, \policy^\game_2) - \simcost \\
        & \iff
            \alpha \utility^\game_1(\policy^\game_1)
                = \alpha \cdot 2 - 1 \\
        & \iff
            \alpha = \frac{1}{2 - \utility^\game_1(\policy^\game_1)} \\
        & \iff
            \utility^\game_1(\policy^\game_1) = 2 - \frac{1}{\alpha}
        .
    \end{align*}
This concludes the whole proof.
\end{proof}
}

\edit{
\newcommand{\NPproblem}{P}
\newcommand{\x}{x}
\newcommand{\y}{y}
\newcommand{\reduction}{\hat \game}
\newcommand{\preReduction}{\widetilde{\game}}
\newcommand{\modifiedRPS}{\widetilde{\textnormal{RPS}}}

\begin{figure*}[!tb]
    \centering
    \includegraphics[width=0.24\textwidth]{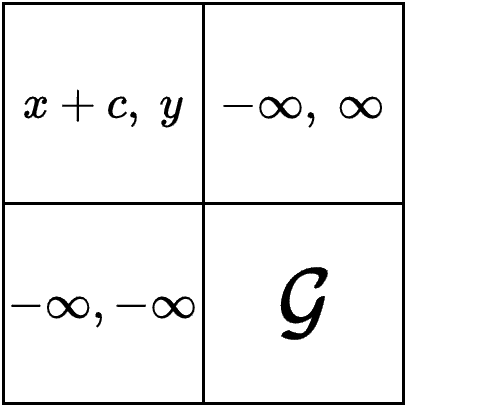}
    \caption{
        The game $\reduction$ which shows that
            determining whether $\game$ has an equilibrium with $\utility_1 \geq \x$
        can be reduced to
            determining whether simulation is guaranteed to lead to a Pareto-improvement.
        \\
        In $\reduction$, the players first simultaneously announce whether they wish to play $\game$ or not,
            with $\game$ being played only if both agree to do so.
        The top-left corner payoff for \Pltwo{}'s is defined as
            $\y := \max_{(\action, \actionOpp) \in \actions^\game} \utility_2(\action, \actionOpp)$,
            which causes voting for playing $\game$ to be a dominant action for \Pltwo{}.
        The top-left corner payoff for \Plone{} is a free parameter of $\reduction$.
    }
    \label{fig:usefulness-is-NP-hard}
\end{figure*}

\usfulnessIsNPhard*

\begin{proof}
First, note that the problem $\NPproblem(\x, \game)$ := ``given an NFG $\game$ and real number $\x$, decide whether $\game$ has a Nash equilibrium with $\utility_1 > \x$'' is known to be NP-hard \cite[Corollary 8]{conitzer2008new}.
To prove the theorem, it therefore suffices,
        for every $\simcost \in \R$,
    to reduce from the complement of $\NPproblem(\x,\game)$ to the problem
        ``given a game $\reduction$, decide whether the game $\reduction_\mathrm{sim}^\simcost$ has a NE with $\policy_1(\simulate) > 0$ that is a Pareto-improvement over every NE of $\reduction$''.
We first consider the case where (A) $\simcost < 0$
    and then (B) extend the solution to $\simcost \geq 0$.

\noindent
A) $\simcost < 0$:
We construct $\reduction$ as follows (\Cref{fig:usefulness-is-NP-hard}).
    First, the players simultaneously announce whether they wish to play $\game$.
    If both agree to play $\game$, they play $\game$.
    If neither agrees,
        \Plone{} receives payoff $\x + \simcost$ and
        \Pltwo{} receives payoff $\y := \max_{(\action, \actionOpp) \in \actions^\game} \utility_2(\action, \actionOpp) + 1$.
        If only \Plone{} agrees, both players receive a very large negative payoff
            (intuitively: $-\infty$;
            formally:
            $
                \utility_\pl
                :=
                \min\{
                    \min_{(\action, \actionOpp) \in \actions}
                        \utility_\pl (\action, \actionOpp),
                    \x + \simcost
                \} - 1
            $).
        If only \Pltwo{} agrees,
            \Plone{} receives a very large negative payoff
            while \Pltwo{} receives a very large positive payoff
            (intuitively: $\infty$;
            formally:
            $
                \utility_2
                :=
                \max_{(\action, \actionOpp) \in \actions}
                    \utility_2 (\action, \actionOpp) + 2 > \y
            $).

First, note that the dominant action in $\reduction$ is for \Pltwo{} to vote for $\game$, causing \Plone{} to do the same.
    Therefore, the equilibria of $\reduction$ can be identified with the equilibria of $\game$, and they yield the same utilities.
Second, note that in $\reduction^\simcost_{\mathrm{sim}}$, \Plone{} will always simulate
    (because $\simcost < 0$ makes this a dominant course of action).
This incentivizes \Pltwo{} to always vote against playing $\game$,
    resulting in utilities
        $\utility_1 = (\x + \simcost) - \simcost = \x$ and
        $\utility_2 = \y$ for \Pltwo{}.
\Pltwo{} will always prefer this over playing $\game$,
which shows that
    introducing simulation to $\reduction$ is guaranteed to lead to a Pareto-improvement
    if and only if $\game$ has no NE with $\utility_1 > \x$.

\noindent
B) $\simcost \geq 0$:
\co{Unimportant, but I don't like using textbf here because it changes the appearance of the mathematical symbold. In principle bold c could be something different from normal c. (Same above.) So maybe one could use underline or so?}
We start with a game $\preReduction$ which is defined identically to $\reduction$ from the case ``$\simcost < 0$'',
    except that when the players play the top and left strategy,
    \co{Wait, what does it mean for both players to simulate? I think this is probably just the payoff of the top left corner?}
    \Plone{} receives utility $\utility_1 := \x - 1$ (instead of $\x+\simcost$).
However, we also consider a second game, $\modifiedRPS$,
    which works as the traditional Rock-Paper-Scissors game with utilities $\{-1, 0, 1\}$,
    except that
        \Plone{}'s payoffs are multiplied by $\simcost + 1$ and
        \Pltwo{}'s payoffs are multiplied by $\nicefrac{1}{2}$.
Finally, game $\reduction$ is the game where the players play the games $\preReduction$ and $\modifiedRPS$ in parallel
    (or one after the other; this does not make a difference).

Because of the inclusion of $\modifiedRPS$, any NE of $\reduction$ will involve \Plone{} simulating with probability $1$.
    (Indeed, simulating causes them to
        lose $\simcost$ utility by paying the simulation cost and
        gaining $\simcost + 1 > \simcost$ utility for being able to always win the $\modifiedRPS$ game.
    Therefore, against any NE strategy of the opponent, \Plone{} will always simulate.)
As before,
    this incentives \Pltwo{} to always vote against playing $\game$,
    which in turn makes \Plone{}'s best response to also vote against playing $\game$.
This results in utilities
    $\utility_1 = (\x - 1) + (\simcost + 1) - \simcost = \x$ and
    \begin{align*}
        \utility_2
        & =
        \y - \nicefrac{1}{2}
        =
        \max_{(\action, \actionOpp) \in \actions^\game}
            \utility_2(\action, \actionOpp)
        + 1
        - \nicefrac{1}{2}
        \\
        & >
        \max_{(\action, \actionOpp) \in \actions^\game} \utility_2(\action, \actionOpp)
        .
    \end{align*}
As in the previous case, this implies that
    introducing simulation to $\reduction$ is guaranteed to lead to a Pareto-improvement
    if and only if $\game$ has no NE with $\utility_1 > \x$.
\end{proof}
}

\edit{
\generalisedTGnoTiebreaking*

\begin{proof}
\newcommand{\gv}{v_2}
\newcommand{\optimalCommitments}{\textnormal{OC}}
\newcommand{\Unif}{\textnormal{Unif}}
\newcommand{\devProb}{q_\simcost}
\newcommand{\ubr}{\textnormal{ubr}}
\newcommand{\utilityAlt}{w}
Let $\game$ be a generalised trust game that does not admit best-response tie-breaking by \Plone{},
    and suppose that $\simcost > 0$ is sufficiently low (to be specified later in the proof).
    We will prove that
        \begin{itemize}
            \item[(A)] either $\simgame^\simcost$ has a NE where $\policy_1(\simulate) = 1$ and \Pltwo{} randomizes over all of their optimal commitments
            \item[(B)] or $\simgame^\simcost$ has a NE where
                \Plone{} mixes between simulating and some ``baseline policy'' that only includes ``universal best responses'' to \Pltwo{}'s optimal commitments and
                \Pltwo{} mixes between (all of) their optimal commitments and some ``deviation policy''
            .
        \end{itemize}
    Note that (B) should be understood in the context of \Cref{lem:limit-NE-structure} --- in particular, the probability of \Pltwo{} deviating will go towards $0$ as $\simcost \to 0_+$.
    Since $\game$ is a generalized trust game,
        the outcome $\utility(\action, \actionOpp)$ corresponding to any optimal commitment $\actionOpp$ and $\action \in \br(\actionOpp)$
        is strict Pareto improvement over any NE of $\game$.
    Moreover (since the set $\NE(\game)$ is compact),
        this further implies that for low enough $\simcost$,
        the same holds for $\utility(\simulate, \actionOpp)$.
    This shows that once we prove that either (A) or (B) holds, we will have proven the conclusion of the whole theorem.

\medskip

\textbf{Notation for (A):}
During the proof, we will use the following notation:
    By $\gv := \max \{ \utility_2(\br, \actionOpp) \mid \actionOpp \in \actions_2 \}$ the \textit{value of the pure-commitment game} given by $\game$.
    By $\optimalCommitments := \{ \actionOpp \in \actions_2 \mid \utility_2(\br, \actionOpp) = \gv \}$, we denote the set of \textit{optimal commitments} of \Pltwo{}.
    For any set of actions $\actionSetOpp \subset \actions_2$, $\ubr(\actionSetOpp)$ will denote the set $\bigcap_{\actionOpp \in \actionSetOpp} \br(\actionOpp)$ of \Plone{}'s actions that work as a \textit{universal best response} to $\actionSetOpp$.
    For a finite set $X$, $\Unif(X)$ denotes the \textit{uniformly random distribution} over $X$

\textbf{Proof of (A):}
If there is no universal best-response for all optimal commitments
    (i.e., if $\ubr(\optimalCommitments) = \emptyset$),
    we define $\policy_1 := \simulate$ and $\policy_2 := \Unif(\optimalCommitments)$.
To see that $\policy$ is a Nash equilibrium, note that
    \Pltwo{} does not have any profitable deviation
    (since \Plone{} is simulating with probability $1$ and all actions from $\optimalCommitments$ give the maximum utility against $\simulate$).
Similarly, no single action of \Plone{} works as a best-response against every $\actionOpp \in \supp(\policy_2)$,
    so we have $\VoI(\Unif(\optimalCommitments)) > 0$.
    By \Cref{lem:VoIandSim}, this means that for $\simcost \leq \VoI(\Unif(\optimalCommitments))$,
        \Plone{} does not have any profitable way of deviating from $\policy$.
This shows that $\policy$ is a NE satisfying (A), and concludes the proof of this case.

\medskip

\textbf{Notation for (B):}
    For any $\policy_1 \in \policies_1$ and $\deviation \in \actions_2 \setminus \optimalCommitments$,
        we use $\deviationRatio_\deviation(\policy_1)$ to denote the \textit{attractiveness ratio} from \Cref{lem:limit-NE-structure}
        \begin{align*}
            \deviationRatio_\deviation(\policy_1)
            = \frac{
                \utility_2(\policy_1, \deviation) - \gv
            }{
                \gv - \utility_2(\br, \deviation)        
            }
        .
        \end{align*}
    We will also consider the following auxiliary game:
    \begin{align*}
        \gameAlt & := (\actions', \utilityAlt), \textnormal{ where } \\
        \actions'_1 & := \ubr(\optimalCommitments) \\
        \actions'_2 & := \actions_2 \setminus \optimalCommitments \\
        \utilityAlt_1(\action, \deviation) & := - \left( \utility_1(\br, \deviation) - \utility_1(\action, \deviation) \right) \\
        \utilityAlt_2(\action, \deviation) & :=
            \deviationRatio_\deviation(\action)
        .
    \end{align*}
    Intuitively,
        $\gameAlt$ is a game where \Pltwo{} aims to equalize \Plone{}'s regrets for playing various actions instead of simulating
        and \Plone{} aims to equalize the attractiveness ratios of \Pltwo{}'s deviations.
    
\smallskip

\textbf{Proof of (B):}
Let $\rho$ be some NE of $\gameAlt$
    and denote $\policy^\baseline_1 := \rho_1$, $\policy^\baseline_2 := \Unif(\optimalCommitments)$, $\policy^\deviate_2 := \rho_2$.
First, note that to prove (B),
    it suffices to show that for sufficiently low $\simcost$, $\simgame^\simcost$ has a NE of the form
    \begin{align*}
        \policy_1 & := (1 - \simProb) \cdot \policy^\baseline_1 + \simProb \cdot \simulate \\
        \policy_2 & := (1 - \alpha \simcost) \cdot \policy^\baseline_2 + \alpha \simcost \cdot \policy^\deviate_2
        ,
    \end{align*}
        where $\simProb \in (0, 1)$ and $\alpha > 0$ do not depend on $\simcost$.
Second, note that to prove the $\simgame^\simcost$ has a NE of this form, it suffices to prove the following claims:
\begin{enumerate}[label={(\roman*)}]
    \item \Plone{} can make \Pltwo{} indifferent between $\policy^\baseline_2$ and $\policy^\deviate_2$ by mixing between $\policy^\baseline_1$ and $\simulate$.
    \item For low enough $\simcost > 0$, \Pltwo{} can make \Plone{} indifferent between $\policy^\baseline_1$ and $\simulate$ by mixing between $\policy^\baseline_2$ and $\policy^\deviate_2$.
    \item The mixing probabilities from (i) and (ii) are uniquely determined by $(\policy^\baseline_1, \policy^\deviate_2)$ and $\simcost$,
        \Plone{}'s probability in fact does not depend on $\simcost$,
        and \Pltwo{}'s probability depends on $\simcost$ as some $\alpha \simcost$, $\alpha > 0$.
    \item For low enough $\simcost$, \Plone{} has no incentive to deviate from the strategy determined by (i-iii).
    \item \Pltwo{} has no incentive to deviate from the strategy determined by (i-iii).
\end{enumerate}
In the remainder of the proof, we focus on showing that (i-v) holds.

Proof of (i):
    To show (i), it suffices to prove that
        \begin{align*}
            \utility_2(\simulate, \policy^\baseline_2) > \utility_2(\simulate, \policy^\deviate_2)
            \ \textnormal{ \& } \ 
            \utility_2(\policy^\baseline_1, \policy^\deviate_2) > \utility_2(\policy^\baseline_1, \policy^\baseline_2)
            .
        \end{align*}
    To see the first inequality, recall first that $\supp(\policy^\baseline_2) = \optimalCommitments$,
        so $\utility_2(\simulate, \policy^\baseline_2) = \gv$.
        Conversely, $\supp(\policy^\deviate_2) \subseteq \actions_2 \setminus \optimalCommitments$,
            so $\utility_2(\simulate, \policy^\deviate_2) < \gv$.
    To see the second inequality,
        recall that $\game$ is a generalised trust game,
            so $\policy^\baseline$ cannot be a NE of $\game$
            (otherwise $\utility_2(\policy^\baseline) = \gv$ could not be a strict improvement to \Pltwo{}'s utility).
        As a result, \Pltwo{} must have \textit{some} $\deviation$ for which $\utility_2(\policy^\baseline_1, \deviation) > \utility_2(\policy^\baseline)$
            (since $\policy^\baseline_1$ is a best response to $\policy^\baseline_2$).
        By definition of $\utilityAlt_2$ and the fact that $(\policy^\baseline_1, \policy^\deviate_2)$, this implies that the same holds for $\policy^\deviate_2$.

Proof of (ii):
    To show (ii), it suffices to prove that for sufficiently low $\simcost > 0$, we have
        \begin{align*}
            \utility_1(\simulate, \policy^\baseline_2) < \utility_1(\policy^\baseline_1, \policy^\baseline_2)
            \ \textnormal{ \& } \ 
            \utility_1(\simulate, \policy^\deviate_2) > \utility_1(\policy^\baseline_1, \policy^\deviate_2)
            .
        \end{align*}
    The first inequality holds because
        every $\action \in \supp(\policy^\baseline_2) \subseteq \ubr(\optimalCommitments)$ is already a best response
        to every $\actionOpp \in \supp(\policy^\baseline_2) = \optimalCommitments$,
        so $\utility_1(\simulate, \policy^\baseline_2) = \utility_1(\policy^\baseline_1, \policy^\baseline_2) - \simcost < \utility_1(\policy^\baseline)$.
    To prove the second inequality,
        first note that any $\deviation \in \supp(\policy^\deviate_2)$ must have $\utility_2(\policy^\baseline_1, \deviation) > \gv$.
            (Indeed, this holds because $\policy^\baseline$ cannot be a NE of $\game$
                (otherwise $\game$ wouldn't be a generalized trust game),
                so one player must have \textit{some} profitable deviation
                --- and it cannot be \Plone{}, since $\policy^\baseline_1$ is already a best response to $\policy^\baseline_2$.
            As a result,
                if some $\deviation' \in \supp(\policy^\deviate_2)$ had $\utility_2(\policy^\baseline_1, \deviation') \leq \gv$,
                \Pltwo{} could unilaterally increase their utility in $\gameAlt$ by deviating to the action with $\utility_2(\policy^\baseline_1, \deviation) > \gv$,
                and $(\policy^\baseline_1, \policy^\deviate_2)$ would not be a NE of $\gameAlt$.)
        This shows that for any $\deviation \in \supp(\policy^\deviate_2)$, $\policy^\baseline_1$ cannot be a best-response to $\deviation$.
            (Indeed, suppose that $\policy^\baseline_1$ was a best-response to $\deviation$.
                This lead to a contradiction, since we would then have both
                    $
                        \utility_2(\br, \deviation)
                        \utility_2(\policy^\baseline_1, \deviation)
                        >
                        \gv
                    $
                    --- because $\game$ admits no best-response utility tiebreaking by \Plone{} ---
                    and
                    $
                        \utility_2(\br, \deviation)
                        < \gv
                    $
                    --- since $\deviation \notin \optimalCommitments$.)
        Consequently, the following quantity $\eta$ is necessarily strictly positive:
            \begin{align*}
                \eta
                :=
                \min \{
                    \utility_1(\br, \deviation) - \utility_1(\policy^\baseline_1, \deviation)
                    \mid
                    \deviation \in \supp(\policy^\deviate_2)
                \}
                .
            \end{align*}
        Finally, by applying the definition of $\eta$, we get the desired inequality for $\simcost < \eta$:
        \begin{align*}
            & \utility_1(\simulate, \policy^\deviate_2) =
                \\
            & =
                \! \! \! \! \sum_{\deviation \in \supp(\policy^\deviate_2)} \! \! \! \! 
                    \policy^\deviate_2(\deviation)
                    \utility_1(\br, \deviation)
                    - \simcost
                \\
            & \geq
                \! \! \! \! \sum_{\deviation \in \supp(\policy^\deviate_2)} \! \! \! \! 
                    \policy^\deviate_2(\deviation) \left(
                        \utility_1(\br, \deviation)
                        - \utility_1(\policy^\baseline_1, \deviation)
                        + \utility_1(\policy^\baseline_1, \deviation)
                        - \simcost
                    \right)
                \\
            & \geq
                \! \! \! \! \sum_{\deviation \in \supp(\policy^\deviate_2)} \! \! \! \! 
                    \policy^\deviate_2(\deviation) \left(
                        \utility_1(\policy^\baseline_1, \deviation)
                        + \eta
                        - \simcost
                    \right)
            \\
            & >
                \! \! \! \! \sum_{\deviation \in \supp(\policy^\deviate_2)} \! \! \! \! 
                    \policy^\deviate_2(\deviation) \utility_1(\policy^\baseline_1, \deviation)
            = \utility_1(\policy^\baseline_1, \policy^\deviate_2)
            .
        \end{align*}
        This concludes the proof of (ii).

Proof of (iii):
    The uniqueness holds because the inequalities in (i) and (ii) are strict.
    The fact that $\policy_1(\simulate)$ does not depend on $\simcost$ follows from the fact that \Pltwo{}'s utility does not depend on it.
    The fact that the probability mass that $\policy_2$ assigns to $\supp(\policy^\deviate_2)$ is of the form $\alpha \simcost$,
        where $\alpha > 0$ does not depend on $\simcost$,
        can be verified by solving the equation
        \begin{align*}
            & \utility_1(\policy^\baseline_1, (1 - \alpha \simcost) \cdot \policy^\baseline_2 + \alpha \simcost \cdot \policy^\deviate_2) \\
            & =
            \utility_1(\simulate, (1 - \alpha \simcost) \cdot \policy^\baseline_2 + \alpha \simcost \cdot \policy^\deviate_2)
            .
        \end{align*}
        Since
            this calculation is exactly the same as the analogous calculation in the proof of \Cref{lem:limit-NE-structure},
            we skip it.

Proof of (iv):
    By (iii), we have
        $\lim_{\simcost \to 0_+} \policy_2 = \Unif(\optimalCommitments)$,
        which implies that
        \begin{align*}
            \utility_1(\policy_1, \policy_2)
            =
                \utility_1(\simulate, \policy_2)
            \overset{\simcost \to 0_+}{\longrightarrow}
                \sum_{\actionOpp \in \optimalCommitments}
                    \frac{\utility_1(\br, \actionOpp)}{|\optimalCommitments|}
            .
        \end{align*}
    As a result, for low enough $\simcost$, any $\action \in \actions_1$ that aims to have
        $
            \utility_1(\action, \policy_2)
            \geq
            \utility_1(\policy_1, \policy_2)
        $
        must be a best response to each of the strategies $\actionOpp \in \optimalCommitments$.
    Since all such actions are contained in $\actions'_1 = \ubr(\optimalCommitments)$,
        it suffices --- to prove (iv) --- to check that \Plone{} has no incentive to deviate towards some $\action \in \ubr(\optimalCommitments)$.
    To prove this, note that for $\action \in \ubr(\optimalCommitments)$, we have
        \begin{align*}
            \utility_1(\action, \policy_2)
            & = (1 - \alpha \simcost) \utility_1(\action, \policy^\baseline_2)
                + \alpha \simcost \utility_1(\action, \policy^\deviate_2)
                \\
            & = (1 - \alpha \simcost) \utility_1(\br, \policy^\baseline_2)
                + \alpha \simcost \sum_{\deviation \in \actions_2}
                        \policy^\deviate_2(\deviation)
                        \utility_1(\action, \policy^\deviate_2)
                \\
            & = (1 - \alpha \simcost) \utility_1(\br, \policy^\baseline_2)
                + \alpha \simcost \utility_1(\br, \policy^\deviate_2)
                \\
                & \ \ \ +
                    \alpha \simcost \sum_{\deviation \in \actions_2} \policy^\deviate_2(\deviation)
                    \left( 
                        \utility_1(\action, \policy^\deviate_2)
                        -
                        \utility_1(\br, \policy^\deviate_2)
                    \right)
                \\
            & = (1 - \alpha \simcost) \utility_1(\br, \policy^\baseline_2)
                + \alpha \simcost \utility_1(\br, \policy^\deviate_2)
                \\
                & \ \ \ +
                    \alpha \simcost \sum_{\deviation \in \actions_2} \policy^\deviate_2(\deviation)
                    \left[
                    - \left( 
                        \utility_1(\br, \policy^\deviate_2)
                        -
                        \utility_1(\action, \policy^\deviate_2)
                    \right)
                    \right]
                \\
            & = (1 - \alpha \simcost) \utility_1(\br, \policy^\baseline_2)
                + \alpha \simcost \utility_1(\br, \policy^\deviate_2)
                \\
                & \ \ \ +
                    \alpha \simcost \utilityAlt_2(\action, \policy^\deviate_2)
            .
        \end{align*}
        However,
            in this formula, only the term $\utilityAlt_2(\action, \policy^\deviate_2)$ depends on the choice of $\action$
            --- and since $(\policy^\baseline_1, \policy^\deviate_2)$ is a NE of $\gameAlt$,
            there can be no $\action$ for which this term would be strictly higher than for the actions from $\supp(\policy^\baseline_1)$.
        This concludes the proof of (iv).

Proof of (v):
    Applying the definition of $\utilityAlt_2$,
        we see the utility of any $\deviation \in \actions_2 \setminus \optimalCommitments$ in $\gameAlt$
        is equal to $\deviation$'s ``attractiveness ratio'' (from \Cref{lem:limit-NE-structure}):
    \begin{align*}
        \utilityAlt_2(\policy^\baseline_1, \deviation)
        & =
            \sum_{\action \in \actions_1}
                \policy^\baseline_1(\action)
                \utilityAlt_2(\action, \deviation)
        \\
        & =
            \sum_{\action \in \actions_1}
                \policy^\baseline_1(\action)
                \frac{
                    \utility_2(\action, \deviation) - \gv
                }{
                    \gv - \utility_2(\br, \deviation)        
                }
        \\
        & =
            \frac{
                \utility_2(\policy^\baseline_1, \deviation) - \gv
            }{
                \gv - \utility_2(\br, \deviation)        
            }
        .
    \end{align*}
    Using the terminology from \Cref{lem:limit-NE-structure}, $\game$ admits no profitable deviation that satisfies $\indiffLabel$ or $\likesSimLabel$
        --- indeed, this is because
            \Pltwo{}'s baseline only consists of optimal commitments (which rules out $\likesSimLabel$) and
            \Plone{}'s baseline consists of best-responses to any action that satisfies $\indiffLabel$
                (so any profitable deviation satisfying $\indiffLabel$ would contradict (D1)).
    Therefore,
        as a result of \Cref{lem:limit-NE-structure} and the choice of $\policy_1(\simulate)$ in (i),
        any profitable deviation of \Pltwo{} would need to have a strictly higher attractiveness ratio than the actions from the support of $\policy^\deviate_2$.
    However,
        the existence of such action would contradict the fact that
        $(\policy^\baseline_1, \policy^\deviate_2)$ is a Nash equilibrium of $\gameAlt$.
This concludes the proof of (v),
    and therefore the proof of the whole theorem as well.
\end{proof}
}

\simInZeroSumGames*

\begin{proof}
For any $\policy \in \NE(\simgame^\simcost)$, we have $\utility_1(\policy) \geq \max_{\rho_1 \in \policies_1(\game)} \utility_1(\rho_1, \policy_2) \geq v$,
which proves the inequality for \Plone{}.

To prove the inequality for \Pltwo{}, note that we can write $\policy_1$ as $\lambda \tilde \policy_1 + (1-\lambda) \simulate$,
    where $\tilde \policy_1$ is a best-response (in $\game$) to $\policy_2$.
Since $\game$ is zero-sum, we have
    $
        \utility_2(\simulate, \policy_2)
        =
        \sum_{\actionOpp}
            \min_{\action} \utility_2(\action, \actionOpp)
        \leq \min_{\action \in \actions_1} \utility_2(\action, \policy_2)
        = \utility_2(\br, \policy_2)
        = \utility_2(\tilde \policy_1, \policy_2)
        \leq - v
        .
    $
\end{proof}

\input{examples}

%% file: examples.tex
\newcommand{\gtnUtil}{150}
\newcommand{\gtnDiff}{300}       
\newcommand{\projectUtil}{100}
\newcommand{\defectUtilA}{200}    
\newcommand{\defectUtilB}{$\epsilon$}
\newcommand{\prisonUtilA}{999}
\newcommand{\numberOfGuesses}{N}
\newcommand{\simcostSpecific}{101}

\begin{figure}[!tb]
    \centering
    \includegraphics[width=0.45\textwidth]{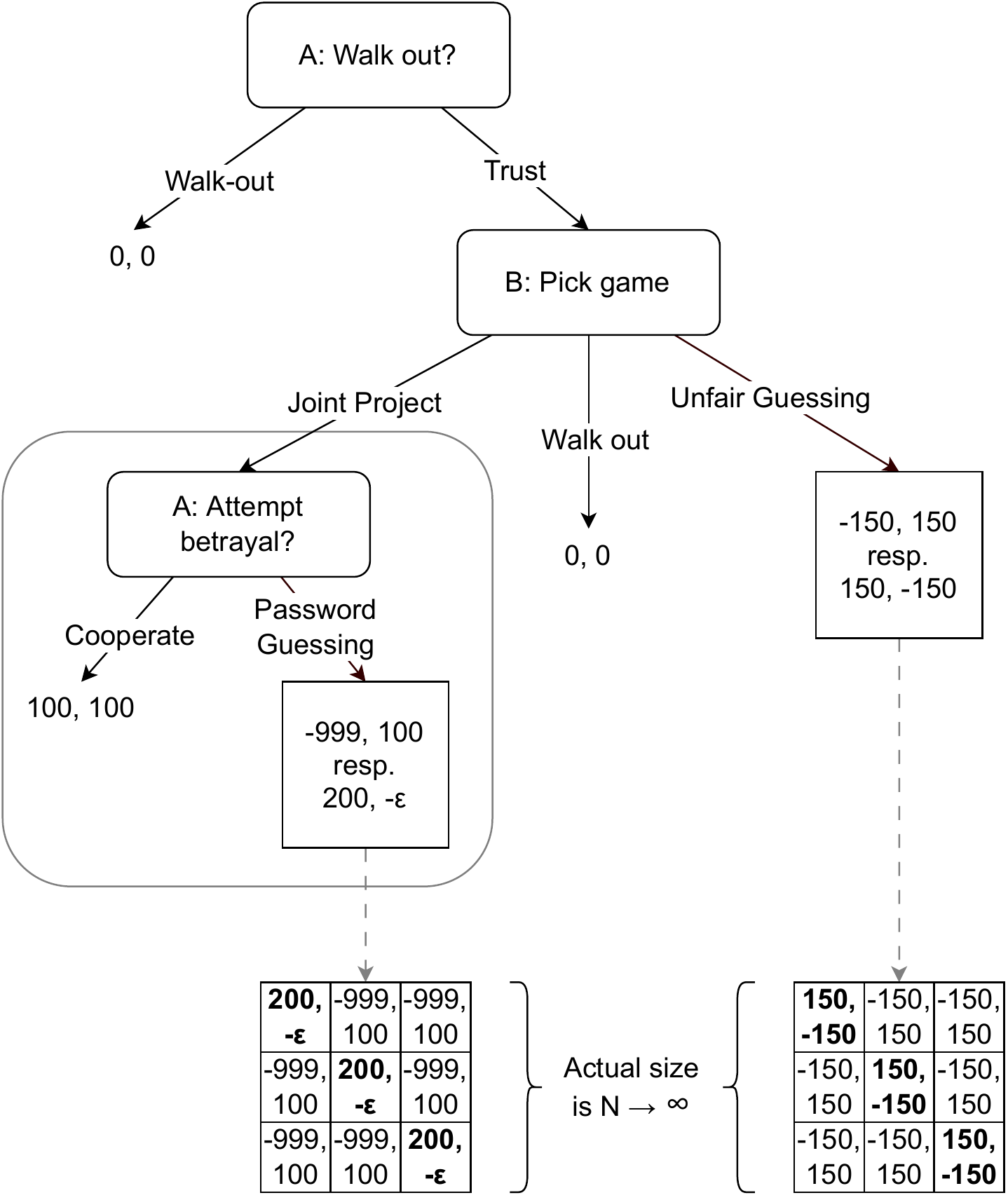}
    \caption{
        A game where both players prefer simulation to be neither cheap nor prohibitively costly.
    }
    \label{fig:nontrivial-simcost}
\end{figure}

\begin{example}[Optimal simulation cost is non-trivial]\label{ex:nontrivial-simcost}
Consider the game $\game$ depicted in \Cref{fig:nontrivial-simcost}.
    First, both Alice and Bob have an option to walk out and not play ($\utility_\Alice = \utility_\Bob = 0$).
    If Alice chooses to trust Bob and play, Bob gets to decide which game to play.
    One option is the Unfair Guessing game (Example~\ref{ex:cheap-exp}), where Alice needs to guess an integer that Bob is thinking, else she ends up transferring $\gtnUtil$ utility to Bob. If she guesses correctly, Bob transfers $\gtnUtil$ to her instead.
        (This game is parametrized by $\numberOfGuesses$, the highest integer that Bob is allowed to pick.
        Since the game is biased in Bob's favor, the corresponding expected utilities converge to $\utility_\Alice = - \gtnUtil, \utility_\Bob = \gtnUtil$ as $\numberOfGuesses \to \infty$.
        Since the precise numbers are not important for our conclusions, we will, for the purpose of this example, treat them as exactly equal to $\pm \gtnUtil$.)
    The other option available to Bob is to play the Joint Project game (\Cref{ex:exp-exp}).
        In this game, Alice can either Cooperate with Bob ($\utility_\Alice = \utility_\Bob = \projectUtil$)
        or attempt to betray him by guessing his password and stealing all his profits.
        A successful betrayal results in utilities $\utility_\Alice = \defectUtilA$, $\utility_\Bob$,
            while an unsuccessful one sends Alice to prison ($\utility_\Alice = -\prisonUtilA$, $\utility_\Bob = \projectUtil$).
        (This game is also parametrized by $\numberOfGuesses$, such that when Bob picks his password uniformly at random, the outcomes converge to $\utility_\Alice = -\prisonUtilA$, $\utility_\Bob = \projectUtil$. To simplify the notation, we treat them as equal to these numbers.)

We first discuss how the game works before simulation enters the picture.
    In this simulation, Alice would prefer to Bob to pick the Joint Project game and she would cooperate if Bob did pick this game.
    However, Bob would rather play the Unfair Guessing game, which is virtually guaranteed to make him better off.
    Realizing this, Alice decides to walks out instead, and the only equilibrium outcome is $\utility_1 = \utility_2 = 0$.

Conversely, if simulation is free,
    the Unfair Guessing game becomes unfavourable to Bob --- he would get $\utility_\Bob = - \gtnUtil$ if he picked it.
    However, simulating would also allow Alice to betray Bob in the Joint Project game, making him worse off than if he didn't play at all.
    As a result, when $\simcost$ is equal to $0$, Bob always walks out and the only equilibrium outcome is $\utility_1 = \utility_2 = 0$, as before.

However, consider the case where $\simcost = \simcostSpecific$,
    such that simulation is
        cheap enough to be justified by the fear of the Unfair Guessing game,
        but expensive enough to not be justified by the greed in the Joint Project game.
    Intuitively, we might hope that this will cause Bob to pick the Joint Project game, in which Alice will cooperate --- and this is mostly what actually ends up happening.
    The only wrinkle is that Alice needs to make the decision about simulation before knowing Bob's choice of game
        --- and if she never simulates, Bob would switch to always selecting the Unfair Guessing game instead.
        As a result, the actual equilibrium (given below) will have Bob sometimes deviating towards Unfair Guessing and Alice sometimes simulating.
            As a by-product, this will sometimes lead to Alice betraying Bob in the Joint Project Game (when she simulates and he doesn't deviate).
        However, even with these drawbacks, the resulting outcome is still much better than the default $\utility_\Alice = \utility_\Bob = 0$.
    Indeed, it is not hard to verify that $\simgame^{\simcostSpecific}$ has an equilibrium where
        Alice simulates with probability $(\gtnUtil - \projectUtil) / (\gtnDiff - \projectUtil) = \nicefrac{1}{4} $ and trusts Bob otherwise,
        while Bob picks the Unfair Guessing game with probability $(\simcost - \projectUtil) / (\projectUtil + \gtnDiff) = \nicefrac{1}{400}$ and selects the Joint Project game otherwise,
        and the resulting utilities are $\utility_\Alice \approx \projectUtil $, $\utility_\Bob = (1 - \nicefrac{1}{4}) \projectUtil = 75$.
    (Note that -- somewhat counterintuitively -- making the Unfair Guessing game riskier for Bob would make the overall outcome \textit{better} for him,
        because Alice would not need to simulate with so high probability to disincentivize deviation.)
        
By adjusting the payoffs in this game, we can obtain examples where the players prefer various values of $\simcost$ that
    are strictly higher than $0$, yet induce equilibria where simulation happens with non-zero probability.
\end{example}

%% file: piecewise_linearity.tex
\section{Proof of Proposition~\ref{prop:piecewise_constant_linear}: Linear Adjustment of P1's Payoffs Have Constant/Linear Effects on NE}\label{sec:app:linearity}

\newcommand{\adjust}{\alpha}
\newcommand{\adjustVec}{\vec \adjust}
\newcommand{\costScale}{\simcost}
\newcommand{\adjustGame}{\game_{\adjustVec}}
\newcommand{\Vertices}{\NE^{\textnormal{ext}}}
\newcommand{\constant}{\gamma}
\newcommand{\bfs}{\beta}
\newcommand{\LPforNESS}{the LP from Lemma~\ref{lemma:NEasLPsolutions}}
\newcommand{\convexWeight}{\lambda}

In the main text, we used the following result:

\piecewiseLinear*

\begin{proof}
With the exception of the claim about $\exceptionPoint_0$ being equal to $0$,
    the result is an immediate corollary the more general \Cref{lem:piecewise-linear-general-weaker} listed below.
The claim about $0$ being a breakpoint and about the non-existence of any breakpoints in $(-\infty, 0)$ immediately follows from \Cref{prop:NE-for-extreme-simcost}\,(i).
\end{proof}

In this section, we prove a more general version of \Cref{prop:piecewise_constant_linear} using the following setting:

\begin{definition}[Auxiliary]
A game with \defword{linearly adjustable payoffs}
    is any pair $(\game, \adjustVec)$ where
        $\game = (\actions, \tilde \utility)$ is a two-player normal-form game and
        $\adjustVec = (\adjust_\action)_{\action \in \actions_1} \in {\actions_1}$ is a vector of \defword{adjustments} for \Plone{}'s actions.
For a \defword{cost-scaling} factor $\costScale \in \R$,
    $\adjustGame^\costScale$ denotes the NFG with
        actions $\actions$ and
        utilities
            $\utility_2 := \tilde \utility_2$,
            $\utility_1(\action, \actionOpp) := \tilde \utility_1(\action, \actionOpp) - \costScale \adjust_\action$.
\end{definition}

\noindent
The connection between this notion and our setting is that
any simulation game $\simgame^\simcost$ can be expressed as
    $\simgame^\simcost = (\game')_{(0, \dots, 0, 1)}^\simcost$,
    where $\game'$ is the original game $\game$
    with one additional \Plone{} action $\simulate$
    that yields utilities $\tilde \utility_\pl(\simulate, \actionOpp) := \utility_\pl(\action, \actionOpp)$
        (where we fix some $\action \in \br(\actionOpp)$ for every $\actionOpp \in \actions_2$).

The piecewise constant/linear phenomenon that we observed on the motivating example of simulation in Trust Game (Figure~\ref{fig:piecewise-linearity}) in fact holds more generally --- for every game with linearly adjustable payoffs.
The goal of this section is to build up to the proof of the following result, which immediately gives our desired result -- \Cref{prop:piecewise_constant_linear} as a corollary:

\begin{lemma}[Games with linearly adjustable payoffs have piecewise constant/linear NE trajectories]
    \label{lem:piecewise-linear-general-weaker}
For every $\game$, there is a finite set of breakpoint values
    $
        -\infty = \exceptionPoint_{-1}
            < \exceptionPoint_0 < \dots < \exceptionPoint_k
            < \exceptionPoint_{k+1} = \infty
    $
    such that the following holds:
    For every $\costScale_0 \in (\exceptionPoint_l, \exceptionPoint_{l+1})$
        and every $\policy^{\costScale_0} \in \NE(\adjustGame^{\costScale_0})$,
    there is a linear mapping
        $\NEtrajectory_2 : \costScale \in [\exceptionPoint_l, \exceptionPoint_{l+1}] \mapsto \policy^\costScale_2 \in \policies_2$
    such that
        $\NEtrajectory_2(\costScale_0) = \policy^{\costScale_0}_2$
        and $(\policy^{\costScale_0}_1, \NEtrajectory_2(\costScale)) \in \NE(\adjustGame^{\costScale})$ for every $\costScale \in [\exceptionPoint_l, \exceptionPoint_{l+1}]$.
\end{lemma}

\subsection{Background: Linear Programming}

Before proceeding with the proof, we recall several results from linear programming.
    (Since these results are standard, they will be given without a proof.
    For a detailed exposition of using LPs for solving normal-form games, see for example \cite{MAS}.)

First, if we can guess the support of a Nash equilibrium, the strategy itself can be found using a linear program:

\begin{observation}[Indifference sets of NE]\label{obs:fixed_support_NE}
Every NE $\policy$ of $\game$ satisfies $\supp(\policy_\pl) \subseteq \br(\policy_\opp)$.
As a result, we can write the set of NE in $\game$ as a (possibly overlapping) union
\begin{align*}
    &
    \NE(\game)
        =
        \bigcup \left\{
            \NE(\game, \actionSubset_1, \actionSubset_2)
            \mid
            \forall \pl :
            \actionSubset_\pl \subseteq \actions_\pl
        \right\},
    \\
    &
    \text{ where }
    \NE(\game, \actionSubset_1, \actionSubset_2)
        :=
        \\
        &
        \phantom{xx}
        \left\{
            \policy \in \NE(\game)
            \mid
            \forall \pl :
            \supp(\policy_\pl) \subseteq \actionSubset_\pl \subseteq \br(\policy_\opp)
        \right\}
        .
\end{align*}
\end{observation}

\begin{lemma}[NE as solutions of LP]\label{lemma:NEasLPsolutions}
For any $\actionSubset_1$, $\actionSubset_2$, the elements of $\NE(\game, \actionSubset_1, \actionSubset_2)$ are precisely (the $\policy$-parts of) the solutions of the following linear program (with no maximisation objective).
\begin{align}
    &
    \forall \pl :
        \nonumber
    \\
    &
    \sum\nolimits_{\action_\pl \in \actionSubset_\pl}
        \policy_\pl (\action_\pl)
        =
        1
    \\
    &
    \utility_\pl (\action_\pl, \policy_\opp) = \constant_\pl
        &
        \text{ for } \action_\pl \in \actionSubset_\pl
    \\
    &
    \utility_\pl (\action_\pl, \policy_\opp) \geq \constant_\pl
        &
        \text{ for } \action_\pl \in \actions_\pl \setminus \actionSubset_\pl
    \\
    &
    \phantom{x}\text{ where the variables satisfy }
        \nonumber
    \\
    &
    \policy_\pl( \action_\pl )
        =
        0
        &
        \text{ for } \action_\pl \in \actions_\pl \setminus \actionSubset_\pl
    \\
    &
    \policy_\pl (\action_\pl) \geq 0
        &
        \text{ for } \action_\pl \in \actionSubset_\pl
    \\
    &
    \constant_\pl \in \R
\end{align}
\end{lemma}

\noindent
Another standard result is that the geometry of the set $\NE(\game, \actionSubset_1, \actionSubset_2)$ can be derived from the LP above:

\begin{lemma}[Geometry of NE]\label{lem:NEformConvexPolygon}
(1) For any $\policy, \policy' \in \NE(\game, \actionSubset_1, \actionSubset_2)$,
    we have $(\policy_1, \policy_2'), (\policy_1', \policy_2) \in \NE(\game, \actionSubset_1, \actionSubset_2)$.

\noindent
(2) $\NE(\game, \actionSubset_1, \actionSubset_2)$ is a convex polytope
    and its vertices are precisely the basic feasible solutions of \LPforNESS{}.
\end{lemma}

\noindent
In light of \Cref{lem:NEformConvexPolygon}, we can denote
    \begin{align*}
        \NE(\game, \actionSubset_1, \actionSubset_2)
        :=
        \NE_1(\game, \actionSubset_1, \actionSubset_2)
        \times
        \NE_2(\game, \actionSubset_1, \actionSubset_2)
        .
    \end{align*}
We also use
    $\Vertices_\pl(\game, \actionSubset_1, \actionSubset_2)$
    to denote the \defword{extremal NE strategies}
    --- i.e., the vertices
    $\NE_\pl(\game, \actionSubset_1, \actionSubset_2)$.

\subsection{Linearity of Simulation Equilibria}

The first observation is that since the utilities of \Pltwo{} do not depend on $\costScale$, their Nash equilibrium strategy of \Plone{} do not need to change either:

\begin{lemma}[WLOG, \Plone{}'s strategy is constant]\label{lem:constancy-WLOG}
Suppose that $\policy$, resp. $\policy'$, is a solution of \LPforNESS{}
    for $\adjustGame^{\symbolPlaceholder}$ and $(\actionSubset_1, \actionSubset_2)$
    for $\costScale$, resp. $\costScale'$.
Then $(\policy_1, \policy'_2)$ is a solution of the LP for $\adjustGame^{\costScale'}$.
\end{lemma}

\begin{proof}
To prove the lemma, it suffices to verify that $(\policy_1, \policy'_2)$ is a feasible solution of \LPforNESS{}.
However, this is trivial once we realise that the utility of \Pltwo{} does not depend on $\simcost$.
\end{proof}

As a result, it only remains to prove the linearity of $\Vertices_2(\adjustGame^\costScale, \actionSubset_1, \actionSubset_2)$
    (and then put all the results together).

\begin{proposition}\label{prop:linearity-extremal-NE}
For every $\actionSubset_1 \subset \actions_1$ and $\actionSubset_2 \subset \actions_2$,
    there is a finite number of breakpoints $- \infty = \exceptionPoint_{-1} < \dots < \exceptionPoint_{k+1} = \infty$,
    such that on any of the intervals $(\exceptionPoint_i, \exceptionPoint_{i+1})$,
    the elements of $\Vertices_2(\adjustGame^\costScale, \actionSubset_1, \actionSubset_2)$ change linearly with $\costScale$.
\end{proposition}

\noindent
Here, ``elements of $\Vertices_2(\adjustGame^\costScale, \actionSubset_1, \actionSubset_2)$ changing linearly'' means that
    (a) for a fixed $i$, there is some $N \geq 0$ such that 
    for every $(\exceptionPoint_i, \exceptionPoint_{i+1})$,
    the set $\Vertices_2(\adjustGame^\costScale, \actionSubset_1, \actionSubset_2)$ has exactly $N$ elements and
    (b) there are linear functions $\NEtrajectory_2^n : (\exceptionPoint_i, \exceptionPoint_{i+1}) \to \policies_2$, $n = 1, \dots, N$,
    such that for every $\costScale \in (\exceptionPoint_i, \exceptionPoint_{i+1})$,
    $
        \Vertices_2(\adjustGame^\costScale, \actionSubset_1, \actionSubset_2)
        = 
        \{
            \NEtrajectory_2^n(\costScale)
            \mid
            n = 1, \dots, N
        \}
        .
    $

\begin{figure*}[htb]
    \centering
    \begin{equation*}
        \left[
        \begin{array}{ccccccccc|c}
        1 & 1 & \cdots & 1 & & \dots & & 0 & \phantom{\shortminus} 0 & 1\\
        &&&&&&&&&\\
        \utility_1(\action_1, \actionOpp_1) & \utility_1(\action_1, \actionOpp_2) & \cdots & \utility_1(\action_1, \actionOpp_m) & & \dots & & 0 & \shortminus 1 & \adjust_{\action_1} \\
        \utility_1(\action_2, \actionOpp_1) & \utility_1(\action_2, \actionOpp_2) & \cdots & \utility_1(\action_2, \actionOpp_m) & & \dots & & 0 & \shortminus 1 & \adjust_{\action_2}\\
        \vdots & \vdots & \vdots & \vdots & & & & \vdots & \phantom{\shortminus} \vdots & \vdots\\
        \utility_1(\action_n, \actionOpp_1) & \utility_1(\action_n, \actionOpp_2) & \cdots & \utility_1(\action_n, \actionOpp_m) & & \dots & & 0 & \shortminus 1 & \adjust_{\action_n}\\
        \utility_1(\action'_1, \actionOpp_1) & \utility_1(\action'_1, \actionOpp_2) & \cdots & \utility_1(\action'_1, \actionOpp_m) & 1 & 0 & \dots & 0 & \shortminus 1 & \adjust_{\action'_1}\\
        \utility_1(\action'_2, \actionOpp_1) & \utility_1(\action'_2, \actionOpp_2) & \cdots & \utility_1(\action'_2, \actionOpp_m) & 0 & 1 &  & 0 & \shortminus 1 & \adjust_{\action'_2}\\
        \vdots & \vdots & \ddots & \vdots & \vdots & & \ddots & \vdots & \phantom{\shortminus}\vdots & \vdots \\
        \utility_1(\action'_{n'}, \actionOpp_1) & \utility_1(\action'_{n'}, \actionOpp_2) & \cdots & \utility_1(\action'_{n'}, \actionOpp_m) & 0 & 0 & \dots & 1 & \shortminus 1 & \adjust_{\action'_{n'}}\\
        \end{array}
        \right]
        ,
    \end{equation*}
    \caption{The matrix form $\LPmatrix \LPvariables^\transpose = \LPrhs$ of the LP \eqref{eq:LP-slack-variables-start}-\eqref{eq:LP-slack-variables-end}, where
        the columns are indexed by
            $
                \LPvariables
                =
                (
                    \policy^\costScale_2(\actionOpp_1), \dots, \policy^\costScale_2(\actionOpp_1),
                    \slackVariable_{\action'_1}, \dots, \slackVariable_{\action'_{n'}},
                    \constant
                )
            $.
        (The numbers $m$, $n$, and $n'$ stand for the size of $\actionSubset_2$, $\actionSubset_1$, and $\actions_1 \setminus \actionSubset_1$ respectively.)
        The additional constraints are $\policy^\simcost_2(\actionOpp_j) \geq 0$ and $\slackVariable_\action' \geq 0$.
        Note that because \Plone{}'s utilities are adjusted independently of \Pltwo{}'s actions, the adjustments $\adjust_\action$ and $\adjust_\actionAlt$ can be moved to right-hand side of the equation.
    }
        \label{fig:matrix-form-of-LP2}
    \end{figure*}

\begin{proof}
Let $(\game, \adjustVec)$ be a game with linearly adjustable payoffs
    and suppose that $\costScale$ is such that there exists some NE in $(\adjustGame^\costScale, \actionSubset_1, \actionSubset_2)$.
    
As the first step, we rewrite the relevant part of the LP from earlier.
    Using \Cref{lemma:NEasLPsolutions}
        (in combination with (1) from \Cref{lem:NEformConvexPolygon}),
        we see that a policy $\policy_2$ lies in $\Vertices_2(\adjustGame^\costScale, \actionSubset_1, \actionSubset_2)$
        if and only if it is a basic feasible solution of the following LP:
    \begin{align}
        & \sum\nolimits_{\actionOpp \in \actionSubset_2}
            \policy_2(\actionOpp)
            = 1
            \\
        & \utility_1(\action, \policy_2) - \costScale \adjust_\action = \constant
            & \textnormal{ for } \action \in \actionSubset_1
            \\
        & \utility_1(\actionAlt, \policy_2) - \costScale \adjust_\actionAlt \geq \constant
            & \textnormal{ for } \actionAlt \in \action_1 \setminus \actionSubset_1
            \\
        & \phantom{x}\textnormal{ where the variables satisfy} \\
        & \policy_2(\actionOpp) = 0
            & \textnormal{ for } \actionOpp \in \actions_2 \setminus \actionSubset_2
            \\
        & \policy_2(\actionOpp) \geq 0
            & \textnormal{ for } \action \in \actionSubset_2
            \\
        & \gamma \in \R
        .
    \end{align}
We turn all of the inequalities into equalities by introducing slack variables
    $\slackVariable_\actionAlt \geq 0$, $\actionAlt \in \actions_1 \setminus \actionSubset_1$:
    \begin{align}
        & \sum\nolimits_{\actionOpp \in \actionSubset_2}
            \policy_2(\actionOpp)
            = 1
                \label{eq:LP-slack-variables-start}
            \\
        & \utility_1(\action, \policy_2) - \costScale \adjust_\action = \constant
            & \textnormal{ for } \action \in \actionSubset_1
            \\
        & \utility_1(\actionAlt, \policy_2) - \costScale \adjust_\action + \slackVariable_\actionAlt = \constant
            & \textnormal{ for } \actionAlt \in \actions_1 \setminus \actionSubset_1
            \\
        & \phantom{x}\textnormal{ where the variables satisfy} \\
        & \policy_2(\actionOpp) = 0
            & \textnormal{ for } \actionOpp \in \actions_2 \setminus \actionSubset_2
            \\
        & \policy_2(\actionOpp) \geq 0
            & \textnormal{ for } \action \in \actionSubset_2
            \\
        & \gamma \in \R
            \label{eq:LP-slack-variables-end}
        .
    \end{align}

\medskip

Second, we rewrite the LP \eqref{eq:LP-slack-variables-start}-\eqref{eq:LP-slack-variables-end} in a matrix form.
    We denote the relevant variables as
        $\LPvariables =
            (
                \policy^\simcost_2(\actionOpp_1), \dots, \policy^\simcost_2(\actionOpp_n),
                \slackVariable_{\actionAlt_1}, \dots, \slackVariable_{\actionAlt_{n'}},
                \gamma
            )
        $,
        where
            $n \ := \lvert \actions_1 \rvert$,
            $n'\ := \lvert \actions_1 \setminus \actionSubset_1 \rvert$.
    In this notation,
        there will be some matrix $\LPmatrix$ and right-hand side $\LPrhs = (1, 0, \dots, 0)$
        for which some $\policy_2$ is a solution of \eqref{eq:LP-slack-variables-start}-\eqref{eq:LP-slack-variables-end}
        if and only if it satisfies
            $\LPmatrix \LPvariables^\transpose = \LPrhs$
            and $\policy^\simcost_2(\actionOpp_j), \slackVariable_\actionAlt \geq 0$.
    However,
        we can additionally use the fact that the adjustment cost $\costScale \adjust_\action$ that \Plone{} pays does not depend on the action of \Pltwo{}.
        This allows us the matrix form depicted in Figure~\ref{fig:matrix-form-of-LP2},
            where all the $\costScale \adjust_\action$-s have been moved to the right-hand side.

\medskip

As the third step,
    we note that
        while the matrix from Figure~\ref{fig:matrix-form-of-LP2} might have linearly dependent rows,
        we can always replace it by a matrix whose rows are linearly independent.
    To see this, note first that
        clearly no row corresponding to one of the actions $\actionAlt \in \actions_1 \setminus \actionSubset_1$
        can be expressed as a linear combination of any other rows,
        because of the $1$-s in the bottom-right corner of the matrix.
    Second,
        it \textit{is} possible that one of the rows corresponding to some $\action \in \actionSubset_1$
        can be expressed as a linear combination of the other rows corresponding to $\actionSubset_1$.
        Suppose that $\action$ is such action and $\lambda_i \in \R$ are the corresponding weights.
        There are two options:
            If $\sum_i \lambda_i \adjust_{\action_i} = \adjust_\action$,
                the condition corresponding to $\action$ can be omitted, since it is already subsumed by the conditions corresponding to $\actionSubset_1 \setminus \{ \action \}$.
            Conversely, if $\sum_i \lambda_i \adjust_{\action_i} \neq \adjust_\action$,
                the system of equations from Figure~\ref{fig:matrix-form-of-LP2} will be unsolvable.
            However, this case is ruled out by our assumption that $\costScale$ is such that $\NE(\game, \actionSubset_1, \actionSubset_2)$ is non-empty.
In summary: For the remainder of the proof, we can assume that
    the rows of the matrix $\LPmatrix$ are linearly independent.

\medskip

Fourth, we identify the basic solution of the LP given by Figure~\ref{fig:matrix-form-of-LP2}.
    For the purpose of this step, denote 
        the set of column-indices of $\LPmatrix$ as
            $
                \indices
                :=
                \actionSubset_2
                \cup \left( \actions_1 \setminus \actionSubset_1 \right)
                \cup \{ \constant \}
            $.
        For a ``basis'' $\basis \subseteq \indices$,
            we use $\LPmatrix_\basis$ to denote the sub-matrix of $\LPmatrix$ consisting of the columns indexed by $\basis$.
        By $\regularBases$, we denote the set of all $\basis$-s for which the sub-matrix $\LPmatrix_\basis$ is regular.
        Finally, for $\basis \in \regularBases $, we denote by $\LPvariables^\basis$ the basic solution corresponding to $\basis$
        --- i.e., the solution of $\LPmatrix \LPvariables^\transpose = \LPrhs$ for which all the variables indexed by $\indices \setminus \basis$ are equal to $0$.
    By definition of a BFS, the basic feasible solutions of the LP are
        precisely all the vectors of the form $\LPvariables^\basis$, $\basis \in \regularBases$.

\medskip

Fifth,
    we show that every basic (not necessarily feasible solution of the LP given by Figure~\ref{fig:matrix-form-of-LP2} changes linearly with $\costScale$.
    To see this,
        note that each basic (not necessarily feasible) solution $\LPvariables^\basis$ can be written as the vector $\LPmatrix_\basis^{-1} \LPrhs$, extended by $0$-s at the indices $\indices \setminus \basis$.
        (Since $\LPmatrix_\basis$ is assumed to be regular, the inverse exists.)
    Since the matrix $\LPmatrix$ does not depend on $\costScale$ and $\LPrhs$ only depends on $\costScale$ linearly, the mapping $\costScale \mapsto \LPvariables^\basis$ is linear.

\medskip

Finally, we conclude the proof.
    To do this, recall that a basic solution $\LPvariables^\basis$ is feasible
        if all of the variables $\policy^\simcost_2(\actionOpp), \slackVariable_\action$ are non-negative.
    For every $\LPvariables^\basis$,
        the set of the values of $\costScale$ for which all of these definitions are satisfied
        is going to be some (possibly empty or trivial) closed interval $[\exceptionPoint^\basis_0, \exceptionPoint^\basis_1]$.
    By taking the set
        $
            \{
                \exceptionPoint^\basis_i
                \mid
                i = 0, 1, \,
                \basis \in \regularBases
            \} \cup \{ -\infty, \infty \}
        $
        and reordering it as an increasing sequence,
        we obtain the desired breakpoint set $\exceptionSet$.
This completes the whole proof
\end{proof}

\begin{proof}[Proof of \Cref{lem:piecewise-linear-general-weaker}]
Let $(\game, \adjustVec)$ be a game with linearly adjustable payoffs.
To get the desired sequence of breakpoints, we take
    -- for every pair $\actionSubset_1 \subseteq \actions_1$ and $\actionSubset_2 \subseteq \actions_2$ --
    some set $\exceptionSet(\actionSubset_1, \actionSubset_2)$ of breakpoints given by \Cref{prop:linearity-extremal-NE}
    and define
    $
        \exceptionSet
        :=
        \bigcup\nolimits_{\actionSubset_1 \subseteq \actions_1, \actionSubset_2 \subseteq \actions_2}
            \exceptionSet(\actionSubset_1, \actionSubset_2)
    $.
    We then enumerate $\exceptionSet$ as a strictly increasing sequence $(\exceptionPoint_i)_i$
To prove our result, let $\costScale \in [\exceptionPoint_i, \exceptionPoint_{i+1}]$,
    and $\policy^{\costScale_0} \in \NE(\adjustGame^{\costScale_0})$.
We finding a trajectory $\NEtrajectory_2$ and that satisfies the conclusion of the lemma.

By \Cref{obs:fixed_support_NE},
    there are some sets $\actionSubset_1$, $\actionSubset_2$ for which $\policy^{\costScale_0} \in \NE(\adjustGame^{\costScale_0}, \actionSubset_1, \actionSubset_2)$.
By \Cref{lem:NEformConvexPolygon} and the subsequent observation,
    there are some
        basic feasible solutions $\bfs^1, \dots, \bfs^N$ (of \LPforNESS)
        and convex combination $\convexWeight_1, \dots, \convexWeight_N$
    such that
        $
            \policy^{\costScale_0}_2
            = \sum\nolimits_{i=1}^N \convexWeight_i \bfs^i_2
        $.
By \Cref{prop:linearity-extremal-NE},
    there are some linear trajectories $\NEtrajectory^1_2, \dots, \NEtrajectory^N_2$
    such that
        $\NEtrajectory_2^i(\costScale_0) = \bfs^i_2$
        and for every $\costScale \in (\exceptionPoint_i, \exceptionPoint_{i+1})$,
            $\NEtrajectory_2^i(\costScale) \in \Vertices_2(\adjustGame^\costScale, \actionSubset_1, \actionSubset_2)$.
    Moreover, by \Cref{lem:constancy-WLOG}, we have
        $(\policy^{\costScale_0}_1, \NEtrajectory_2^i(\costScale) ) \in \NE(\adjustGame^\costScale, \actionSubset_1, \actionSubset_2)$
        for every $\costScale \in (\exceptionPoint_i, \exceptionPoint_{i+1})$.
By the convexity of NE (\Cref{lem:NEformConvexPolygon}),
    $\NEtrajectory_2 := \sum\nolimits_{i=1}^N \NEtrajectory_2^i$
    is a linear trajectory for which
        $\NEtrajectory_2(\costScale_0) = \policy^{\costScale_0}$
        and for every $\costScale \in (\exceptionPoint_i, \exceptionPoint_{i+1})$,
            $(\policy^{\costScale_0}, \NEtrajectory_2^i(\costScale)) \in \NE(\adjustGame^\costScale, \actionSubset_1, \actionSubset_2)$.
Moreover, since the utilities in $\adjustGame^\costScale$ depend continuously on $\costScale$,
    this also implies that
    $
        (\policy^{\costScale_0}, \NEtrajectory_2^i(\costScale))
        \in
        \NE(\adjustGame^\costScale, \actionSubset_1, \actionSubset_2)
    $
    for $\costScale \in \{ \exceptionPoint_i, \exceptionPoint_{i+1}\}$.
(By \Cref{obs:fixed_support_NE},) this concludes the whole proof.
\end{proof}